\documentclass[lettersize,onecolumn]{IEEEtran}
\usepackage{amsmath,amsfonts}
\usepackage{algorithmic}
\usepackage{algorithm}
\usepackage{array}
\usepackage[caption=false,font=normalsize,labelfont=sf,textfont=sf]{subfig}
\usepackage{textcomp}
\usepackage{stfloats}
\usepackage{url}
\usepackage{verbatim}
\usepackage{graphicx}
\usepackage{cite}
\usepackage{xcolor,bm,physics}
\usepackage[inline]{enumitem}

\newtheorem{defi}{Definition}
\newtheorem{exam}{Example}
\newtheorem{lem}{Lemma}
\newtheorem{theo}{Theorem}

\newtheorem{coro}{Corollary}
\newtheorem{rem}{Remark}

\newcommand{\R}[1]{\textcolor{red}{#1}}
\colorlet{RED}{red}
\newcommand{\B}[1]{\textcolor{blue}{#1}}

\newcommand{\SET}[1]{\mathbb{#1}}
\newcommand{\NNR}{\SET{R}_{\ge 0}}
\newcommand{\Z}{\SET{Z}}
\newcommand{\N}{\SET{N}}
\newcommand{\C}{\SET{C}}

\newcommand{\setint}[1]{[\hspace{-0.5mm}[#1]\hspace{-0.5mm}]}
\newcommand{\subint}[2]{\tbinom{\setint{#1}}{#2}}

\newcommand{\mat}[1]{\mathbf{#1}}

\newcommand{\pure}[2]{\C^{#1 \otimes #2}}
\newcommand{\state}[2]{S(\pure{#1}{#2})}
\newcommand{\QU}[1]{\mathcal{#1}}

\newcommand{\Corr}[1]{#1}

\begin{document}

\title{Insertion Correcting Capability for Quantum Deletion-Correcting Codes}
\author{Ken Nakamura,~\IEEEmembership{Graduate Student Member,~IEEE},
  \Corr{Kaito Sato},
 Takayuki Nozaki,~\IEEEmembership{Member,~IEEE},
 \thanks{
   Manuscript received Nov. 29, 2021; revised ** **, 20**.
   The work of K.~Nakamura
   was
   supported by JSPS KAKENHI Grant Number 26KJ1754.   
   The work of T.~Nozaki
   was
   supported by JSPS KAKENHI Grant Number 22K11905.    
  }
  \thanks{
    K.~Nakamura, \Corr{K.~Sato}, and T.~Nozaki
    are
    with
    Dept.\ of Informatics, Yamaguchi University, Yamaguchi 753-8512, JAPAN.
    e-mail:
    f006wbw@yamaguchi-u.ac.jp,
    g024vbu@yamaguchi-u.ac.jp,
    tnozaki@yamaguchi-u.ac.jp.
  }
}



\maketitle

\begin{abstract}
\Corr{It is conjectured that 
a quantum code $C$ corrects $t$ deletion errors
if and only if 
$C$ corrects a total of $t$ insertion and deletion errors.
Although this conjecture has been proved only under certain conditions,
it has not been proved in general.
In this paper, 
we prove 
this conjecture for general quantum codes.
Specifically, 
we show that 
a quantum state after an insertion error is 
represented as 
a product state of the original multi-qudit system and the inserted quantum state.
Then,
we prove this conjecture
by combining the previous results and this property.
In addition, 
we define generalized insertion errors by 
removing the quantum channel requirement from insertion errors
and
present their properties.}
\end{abstract}


\section{Introduction}
\IEEEPARstart{I}{n} classical coding theory, 
a classical insertion error adds a new symbol to a sequence
and
a classical deletion error removes a symbol from a sequence.
It is known that 
a classical code $C$ corrects $t$ classical deletion errors
if and only if 
$C$ corrects a total of $t$ classical insertion and deletion errors \cite{levenshtein1966binary}.
The correcting capability for classical insertion and deletion errors
is characterized by the insdel distance \cite{levenshtein1966binary}.

In quantum coding theory, 
insertion and deletion errors are defined as classical coding theory;
namely,
an {\it insertion} error adds a new qudit to a multi-qudit system
and
a {\it deletion} error removes a qudit from a multi-qudit system.
\Corr{In particular,
we call
an insertion error
where there is no entanglement 
between the original multi-qudit system and the inserted state a {\it block separable (BS) insertion} error.
In addition,
we call
a BS insertion error 
where the inserted state is a product state {\it separable insertion error}.}
It is conjectured that 
a quantum code $C$ corrects $t$ deletion errors
if and only if 
$C$ corrects a total of $t$ insertion and deletion errors
\cite{matsumoto2021survey,hagiwara2025introduction}.

This conjecture has been proved only under certain conditions.
Shibayama and Ouyang \cite{shibayama2021equivalence} showed that 
the Knill-Laflamme (KL) condition for $t$ deletion errors 
is equivalent to the KL condition for a total of $t$ deletion and separable insertion errors.
It is known that 
insertion errors are limited to \Corr{BS insertion errors} 
if the original multi-qudit system is pure \cite{kimura2011general,shibayama2022equivalence}.
\Corr{From these results,
it follows that 
a quantum code $\QU{C}$ corrects a single deletion error
iff
$\QU{C}$ corrects a single insertion error \cite{shibayama2025necessary}.
In addition,
Bulled and Ouyang \cite{bulled2026equivalence} 
analyzed composite insertion and deletion errors 
using a commutation-type relation between insertion and deletion channels.
Then,
they showed that
a permutation-invariant (PI) code $\QU{C}_{\text{PI}}$ \cite{Ouyang:2014lfv} corrects $t$ deletion errors
iff
$\QU{C}_{\text{PI}}$ corrects $t$ insertion errors.
Moreover,
several quantum deletion-correcting codes \cite{ouyang2021permutation,hagiwara2020four,matsumoto2022constructions,hagiwara2023quantum,hagiwara2025quantum}
correct insertion errors \cite{bulled2025angular,hagiwara2021four,nakamura2025multiple,sasaki2024insertion}.}
However,
this conjecture has not been proved for general quantum codes.

\Corr{In this paper,
we consider two type of insertion errors: 
insertion error with/without quantum channel requirement.
In this paper,
we call insertion error with quantum channel requirement {\it restricted insertion} errors 
or simply {\it insertion errors}
and
insertion error without quantum channel requirement {\it generalized insertion} errors.
Here,
when a quantum channel can be decomposed into deletions and insertions of qudits,
a generalized insertion error can represent the insertion part of the channel.
For example,
although a permutation of qudits is a quantum channel that
can be regarded as a composite error consisting of deletions and insertions,
the insertion part does not satisfy the quantum channel requirement.}

\Corr{This paper studies the insertion-correcting capability
for both types of insertions.
For insertion errors,
this paper shows that this conjecture holds for general quantum codes;
namely,
a quantum code $C$ corrects $t$ deletion errors
if and only if 
$C$ corrects a total of $t$ insertion and deletion errors.
To prove this conjecture,
we first show that
insertion errors are limited to BS insertion errors.
By combining this result with the result in \cite{shibayama2021equivalence},
we prove this conjecture.}

\Corr{For generalized insertion errors,
this paper shows that the analogous equivalence does not hold for general quantum codes.
Specifically,}
a quantum code $\QU{C}$ corrects 
a total of $t$ deletion and generalized insertion errors
if 
$\QU{C}$ corrects $t$ deletion errors,
whereas the converse does not hold.
In addition,
we define the quantum insdel distance
and
characterize
the correcting capability for deletion and generalized insertion errors
by this distance.

The rest of the paper is organized as follows.
Section~\ref{sec:pre} introduces the notations used throughout the paper.
\Corr{Section~\ref{sec:ins_only_BSins} shows that
insertion errors are limited to BS insertion errors.
Section~\ref{sec:KL_BSins=del} proves this conjecture.}
Section~\ref{sec:gene_ins_prop} defines\Corr{ generalized insertion errors}
and presents their properties,
and
Section~\ref{sec:gene_indel_correct} presents
the correcting capability for \Corr{generalized insertion} and deletion errors.
In addition,
Section~\ref{sec:Qindel_dis} defines quantum insdel distance 
and characterizes the correcting capability.
\Corr{Section~\ref{sec:concl} concludes this paper.}

\section{Preliminary}
\label{sec:pre}
This section gives the notations used throughout the paper.
Moreover,
this section introduces
quantum states, deletion errors, and insertion errors.

\subsection{Notations}
Let $\Z$, $\Z^+$, $\N$, $\NNR$, and $\C$ be
the sets of all integers, positive integers, nonnegative integers,
nonnegative real numbers, and complex numbers, respectively.
For $l \in \Z^+$,
define $\Z_l := \{ 0,1,...,l-1 \}$.
For a given set $S$,
let $|S|$ be the cardinality of $S$
and
$S^n$ be the set of all vectors of length $n$ over $S$.
For $n, m \in \mathbb{Z}$,
define
\begin{align*}
  \setint{m,n}
  &:=
  \{ i \in \Z \mid m \leq i \leq n \},\\
  \setint{n}
  &:=
  \setint{1,n},\\
  \subint{n}{m}
  &:=
  \{ S \subset \setint{n} \mid |S|=m \}.
\end{align*}
The complex conjugate and modulus of $c \in \C$ are written as $c^*$ and $|c|$,
respectively.

\Corr{The tensor product of vector spaces $V_1$ and $V_2$ is denoted by
$V_1 \otimes V_2$ and
for a vector space $V$,
$\underbrace{V \otimes V \otimes \cdots \otimes V}_{n\text{ times}}$
is denoted by $V^{\otimes n}$.
In addition,
we write $\C^{l\otimes n}$ for $(\C^l)^{\otimes n}$.
For $\bm{v}_1 \in V_1$ and $\bm{v}_2 \in V_2$,
their tensor product is also denoted by
$\bm{v}_1 \otimes \bm{v}_2$.
In this paper,
we use the Kronecker product 
to represent the tensor product of two vectors or two matrices,
i.e.,}
for an $m \times n$ matrix $\mat{M}_1=(m_{i,j})$ and a $p \times q$ matrix $\mat{M}_2$,
\begin{align*}
  \mat{M}_1 \otimes \mat{M}_2
  :=
  \begin{pmatrix}
    m_{1,1}\mat{M}_2 & m_{1,2}\mat{M}_2 & \cdots & m_{1,n}\mat{M}_2\\
    m_{2,1}\mat{M}_2 & m_{2,2}\mat{M}_2 & \cdots & m_{2,n}\mat{M}_2\\
    \vdots & \vdots & \ddots & \vdots\\
    m_{m,1}\mat{M}_2 & m_{m,2}\mat{M}_2 & \cdots & m_{m,n}\mat{M}_2\\
  \end{pmatrix}.
\end{align*}

For $x \in \mathbb{Z}_l$,
let $\ket{x} \in \C^l$ be a column vector
whose $(x+1)$th component is $1$ and whose other components are $0$.
For $\bm{x} = (x_i) \in \mathbb{Z}_l^n$,
define
\begin{align*}
  \ket{\bm{x}}
  :=
  \ket{x_1} \otimes \ket{x_2} \otimes \cdots \otimes \ket{x_n} 
  \in
  \pure{l}{n}.
\end{align*}
\Corr{A unit vector is denoted by $\ket{\cdot}$.}
Define $\bra{\phi} := \ket{\phi}^{\dagger}$.

\Corr{For a vector space $V$,
a linear map $f : V\to V$ is called a {\it linear operator} on $V$,
and the set of all linear operators on $V$ is denoted by $\QU{L}(V)$.
In particular,
we denote $\QU{L}( \C^n )$ by $M(n)$.}
For $\mat{M} \in M(n)$,
let $\text{tr}(\mat{M})$ be
the trace of $\mat{M}$
and
let $\mat{M}^{\dagger}$ be
the adjoint of $\mat{M}$.
Let $\mat{I}_n$ be
the identity matrix of order $n$.
A matrix $\mat{M} \in M(n)$ is unitary
if $\mat{M}^\dagger \mat{M} = \mat{I}_n$,
Hermitian
if $\mat{M}^\dagger = \mat{M}$.
A Hermitian matrix $\mat{M} \in M(n)$ is positive semidefinite
if $\bm{x}^\dagger \mat{M} \bm{x} \ge 0$
for all column vectors $\bm{x} \in \C^n$.
The notation $\mat{M} \ge 0$
means that
the Hermitian matrix $\mat{M}$
is
positive semidefinite.
Note that
given
any orthonormal basis $\{ \ket{ \bm{x}_L } \}_{ \bm{x} \in \Z_l^n }$ for $\pure{l}{n}$,
$\mat{M} \in M(l^n)$ is represented by
\begin{align}
  \mat{M}
  =
  \sum_{\bm{x},\bm{y} \in \Z_l^n}
  c_{\bm{x},\bm{y}}
  \ket{\bm{x}_L}\bra{\bm{y}_L}
  ,~~(c_{\bm{x}, \bm{y}} \in \C).
  \label{eq:orth_basis_xL}
\end{align}
In particular, 
\begin{align}
  \mat{M}
  =
  \sum_{\bm{x},\bm{y} \in \Z_l^n}
  c_{\bm{x},\bm{y}}
  \ket{\bm{x}}\bra{\bm{y}}
  ,~~(c_{\bm{x}, \bm{y}} \in \C).
  \label{eq:orth_basis_x}
\end{align}

\subsection{Quantum States and Quantum Channels \cite{lidar2013quantum,nielsen2016quantum}}
\Corr{The set of all density matrices $\mat{M} \in \QU{L}(\QU{H})$ on a Hilbert space $\QU{H}$
is denoted by $S(\QU{H})$,}
i.e.,
\begin{align*}
  S(\QU{H})
  :=
  \{
  \rho \in \QU{L}(\QU{H})
  \mid
  \rho \ge 0,
  \text{tr}(\rho) = 1
  \}.
\end{align*}
For simplicity,
let $\state{l}{0} := \Big\{ \begin{pmatrix} 1 \end{pmatrix} \Big\}$.
Define $\state{l}{\N} := \bigcup_{n \in \N} \state{l}{n}$.

By the spectral decomposition,
$\rho \in \state{l}{n}$ is represented by
\begin{align}
  \rho
  &=
  \sum_{ \bm{x} \in \Z_l^n }
  p_{ \bm{x} } \ket{ \bm{x}_L } \bra{ \bm{x}_L },
  \label{eq:spec}
\end{align}
where
$\{ \ket{ \bm{x}_L } \}_{ \bm{x} \in \Z_l^n }$
is an orthonormal basis for $\pure{l}{n}$
and
$p_{ \bm{x} } \in \NNR$ satisfies
$\sum_{ \bm{x} \in \Z_l^n } p_{ \bm{x} } = 1$.
A quantum state $\rho$ is pure
if $\rho$ has rank $1$.
Otherwise $\rho$ is mixed.
When $\rho \in \state{l}{n}$ is pure,
there exists a vector $\ket{\phi} \in \pure{l}{n}$
such that $\rho = \ket{\phi}\bra{\phi}$.
Then,
we denote $\rho$ by $\ket{\phi}$.

\Corr{For two finite-dimensional Hilbert spaces $\QU{H}_A, \QU{H}_B$,
a linear map $\QU{E} : \QU{L}(\QU{H}_A) \to \QU{L}(\QU{H}_B)$ is called a {\it quantum channel}
if 
there exists a set $\{ \mat{E}_{\alpha} \}_{\alpha}$
of matrices $\mat{E}_{\alpha} : \QU{H}_A \to \QU{H}_B$
such that
$\sum_{\alpha} \mat{E}_{\alpha}^\dagger \mat{E}_{\alpha}$ equals the identity matrix
and
\begin{align*}
  \forall
  \mat{M} \in \QU{L}(\QU{H}_A)
  ~~~
  \QU{E}(\mat{M})
  =
  \sum_{\alpha} \mat{E}_{\alpha} \mat{M} \mat{E}_{\alpha}^\dagger.
\end{align*}
The matrix $\mat{E}_{\alpha}$ is called a {\it Kraus operator} of $\QU{E}$
and
the set $\{ \mat{E}_{\alpha} \}_{\alpha}$ is called a {\it Kraus set} of $\QU{E}$.
For quantum channels $\QU{E}_1$ and $\QU{E}_2$,
we call
the quantum channel such in which $\QU{E}_2$ occurs after $\QU{E}_1$ a {\it composite error}
and
denote it by $\QU{E}_2 \circ \QU{E}_1$.}

\Corr{A vector space $\QU{C} \subset \C^{l \otimes n}$
is 
called a {\it quantum code}.
We say that $\QU{C}$ is an {\it $\QU{E}$-correcting code}
if there exists a {\it recovery operator} $\QU{R}$
such that $\QU{R} \circ \QU{E}(\ket{\phi}\bra{\phi}) = \ket{\phi}\bra{\phi}$ for any $\ket{\phi} \in \QU{C}$.
The {\it Knill--Laflamme (KL) condition} \cite{knill1997theory}
gives a necessary and sufficient condition
for quantum error correction.
Note that 
the following theorem holds
even if $\mat{E}_{\alpha}$ is not square \cite{shibayama2025necessary}.}

\Corr{\begin{theo}\cite{knill1997theory}
  Let $\QU{C}$ be a quantum code spanned by $\ket{0_L}$, $\ket{1_L}$, ..., $\ket{(l-1)_L}$.
  Assume that $\QU{E}$ has a Kraus set $\{ \mat{E}_{\alpha} \}_{\alpha}$.
  Then,
  $\QU{C}$ is an $\QU{E}$-correcting code
  if and only if
  there exist $c_{\alpha,\beta} \in \C$
  such that 
  for any $\alpha,\beta$ and $i,j \in \setint{0,l-1}$,
  \begin{align}
    \bra{i_L} \mat{E}_{\alpha}^{\dagger} \mat{E}_{\beta} \ket{j_L}
    =
    c_{\alpha,\beta}
    \delta_{i,j},
    \label{eq:KL_condition}
  \end{align}
  where
  $\delta_{i,j}$ is the Kronecker delta.
  Equation~\eqref{eq:KL_condition} is called the {\it KL condition} for $\QU{E}$.
\end{theo}}

\subsection{Deletion and Insertion Errors}
\label{subsec:def_insdel}
A deletion error
removes
a qudit from a multi-qudit system.
For $\rho \in \state{l}{n}$ given in Eq.~\eqref{eq:orth_basis_x} and
$p \in \setint{n}$,
the partial trace $\text{Tr}_p : \state{l}{n} \to \state{l}{(n-1)}$ is defined as 
\begin{align*}
  \text{Tr}_p(\rho)
  :=
  \sum_{ \bm{x}=(x_i), \bm{y}=(y_i) \in \Z_l^n } 
  c_{\bm{x},\bm{y}} 
  \text{tr}( \ket{x_p}\bra{y_p} )
  \ket{ \bm{x}_{ \neg p } }\bra{ \bm{y}_{ \neg p } },
\end{align*}
where
\begin{align*}
  \bm{x}_{ \neg p } := ( x_1, ..., x_{p-1}, x_{p+1}, ..., x_n ).
\end{align*}
Then,
for $\rho \in \state{l}{n}$ and
a set of deletion indices
$P = \{ p_1,p_2,...,p_s \} \in \subint{n}{s}$ with $p_1 < p_2 < \cdots < p_s$,
define the $s$-deletion error $\QU{D}_P : \state{l}{n} \to \state{l}{(n-s)}$ by
\begin{align*}
  \QU{D}_P(\rho)
  :=
  \text{Tr}_{p_1} \circ \text{Tr}_{p_2} \circ \cdots \circ \text{Tr}_{p_s}(\rho),
\end{align*}
\Corr{For $\ket{\bm{\Psi}} = \ket{\psi_1\psi_2\cdots\psi_s} \in \C^{l\otimes s}$,
the Kraus operator $\mat{D}^{n-s}_{P,\bra{\bm{\Psi}}} : \C^{l \otimes n} \to \C^{l \otimes (n-s)}$
\cite{shibayama2021equivalence}
for a deletion error is defined by
\begin{align*}
  \mat{D}^{n-s}_{P,\bra{\bm{\Psi}}}
  :=
  \mat{B}_1 \otimes \mat{B}_2 \otimes \cdots \otimes \mat{B}_n,
\end{align*}
where
\begin{align*}
  \mat{B}_j
  &:=
  \begin{cases}
    \bra{\psi_i} &(\text{if~}j = p_i \in P)\\
    \mat{I}_l &(\text{if~}j \notin P)
  \end{cases}.
\end{align*}}

An insertion error adds a new qudit to a multi-qudit system.
An insertion error is defined as a quantum channel $\QU{I}_Q$
satisfying $\QU{D}_Q \circ \QU{I}_Q(\rho) = \rho$ 
for any $\rho \in \state{l}{n}$ and $Q \in \subint{n+t}{t}$
\cite{shibayama2022equivalence,hagiwara2025introduction}.
\Corr{Note that this definition includes the insertion errors considered in
\cite{shibayama2021equivalence,shibayama2025necessary,bulled2026equivalence}.
In particular,
an insertion error is called a {\it separable insertion error}
if there is no entanglement
between $\rho$ and the inserted quantum states
or among the inserted quantum states,
e.g.,
$\rho \mapsto \rho \otimes \sigma_1 \otimes \sigma_2 \otimes \cdots \otimes \sigma_t$,
($\sigma_i \in S(\C^l)$).
For $\ket{\bm{\Psi}} = \ket{\psi_1\psi_2\cdots\psi_t} \in \C^{l\otimes t}$,
the Kraus operator $\mat{I}^n_{Q,\ket{\bm{\Psi}}}$ \cite{shibayama2021equivalence} for a separable insertion error is defined as
\begin{align}
  \mat{I}^n_{Q,\ket{\bm{\Psi}}}
  :=
  \mat{A}_1 \otimes \mat{A}_2
  \otimes \cdots \otimes \mat{A}_{n+t},
  \label{eq:S_ins_Kraus}
\end{align}
where
\begin{align*}
  \mat{A}_j
  &:=
  \begin{cases}
    \ket{\psi_i}
      &(\text{if } j=q_i\in Q),\\
    \mat{I}_l
      &(\text{if } j\notin Q).
  \end{cases}
\end{align*}}

\Corr{In this paper,
an insertion error is called a {\it block-separable (BS) insertion}
if there is no entanglement
between $\rho$ and the inserted quantum state,
e.g.,
$\rho \mapsto \rho \otimes \sigma$,
($\sigma \in S(\C^{l\otimes t})$).
The following theorem \cite{kimura2011general} shows that 
a quantum state of a composite system is a product state 
if the state of one of its subsystems is pure.}

\Corr{\begin{theo}[\cite{kimura2011general}]
  \label{theo:after_ins_gene}
  For systems $A$ and $B$,
  let $\QU{H}_A$ and $\QU{H}_B$ be finite-dimensional Hilbert spaces of $A$ and $B$,
  respectively.
  In addition,
  let $\text{Tr}_B$ be the deletion of all qudits in $B$.
  Then,
  for a state $\rho \in S(\QU{H}_A\otimes\QU{H}_B)$ of a composite system $A \times B$,
  if $\Tr_B(\rho) \in S(\QU{H}_A)$ is pure,
  there exists $\sigma_B \in S(\QU{H}_B)$ such that
  \begin{align*}
    \rho
    =
    \Tr_B(\rho) \otimes \sigma_B
  \end{align*}
\end{theo}}

\Corr{The following theorem is a specialization of Theorem~\ref{theo:after_ins_gene}.
In other words,
insertion errors applied to pure states are limited to BS insertions
\cite{shibayama2022equivalence}.}

\begin{theo}[\cite{shibayama2022equivalence}]
  \label{theo:after_ins_gene_ver_shiba}
  \Corr{For a permutation $\tau$ on $\setint{n}$,
  let $\mat{U}_{\tau}$ denote the unitary matrix
  that shifts the $i$-th qudit in $\rho \in \state{l}{n}$
  to the $\tau(i)$-th position,
  i.e.,
  \begin{align*}
    \mat{U}_{\tau} \ket{x_1x_2...x_n}
    =
    \ket{x_{\tau^{-1}(1)} x_{\tau^{-1}(2)} \cdots x_{\tau^{-1}(n)}}.
  \end{align*}}
  In addition,
  for $t,n \in \mathbb{Z}^+$ and
  $Q = \{ q_1, q_2, ..., q_t \} \in \subint{n+t}{t}$ with $q_1 < q_2 < \cdots < q_t$,
  let $\tau^Q$ be the permutation on $\setint{n+t}$
  such that $\tau^Q(n+i)=q_i$ for $i \in \setint{t}$ and
  $j < k \Rightarrow \tau^Q(j) < \tau^Q(k)$ for $j,k \in \setint{n}$.
  Then,
  for a given pure state $\ket{\phi}\bra{\phi} \in \state{l}{n}$,
  \begin{align}
    \QU{I}_{Q}( \ket{\phi}\bra{\phi} )
    =
    \Corr{\mat{U}_{\tau^Q} (\ket{\phi}\bra{\phi} \otimes \pi) \mat{U}_{\tau^Q}^\dagger}
    \label{eq:pure_ins}
  \end{align}
  holds
  for some $\pi \in \state{l}{t}$.
\end{theo}

\Corr{Theorems~\ref{theo:after_ins_gene} and \ref{theo:after_ins_gene_ver_shiba} mean that
an insertion error does not generate entanglement
between the original multi-qudit system and the inserted quantum system
if 
the original multi-qudit system is pure.
On the other hand,
the inserted qudits may be entangled with each other.
In Section~\ref{sec:ins_only_BSins},
we use Theorems~\ref{theo:after_ins_gene} and \ref{theo:after_ins_gene_ver_shiba} to extend 
the non-entanglement property 
to any quantum state.}

\Corr{\begin{rem}
  A quantum channel that increases the number of qudits
  is not necessarily an insertion error as defined above.  
  For example,
  one such error splits the photon representing a qudit
  into two photons \cite{takekoshi2024quantum},
  e.g.,
  \begin{align*}
    &\ket{\phi_1}\bra{\phi_1} \otimes 
    \ket{\phi_2}\bra{\phi_2} \otimes \cdots \otimes 
    \ket{\phi_n}\bra{\phi_n}\\
    &\mapsto
    \rho_1 \otimes \rho_2 \otimes 
    \ket{\phi_2}\bra{\phi_2} \otimes \cdots \otimes 
    \ket{\phi_n}\bra{\phi_n}
  \end{align*}
  We do not consider it in this paper
  since 
  such an error can be represented by 
  a composite error consisting of a single deletion error and two insertion errors.
\end{rem}}

\Corr{An error in which $t$ separable insertion errors occur
after $s$ deletion errors
is called an {\it $(s,t)$-insdel} error.
Using the KL condition \cite{knill1997theory},
the following theorem has been proved \cite{shibayama2021equivalence}.}

\Corr{\begin{theo}\cite{shibayama2021equivalence}
  \label{theo:shibayama2021equivalence}
  The KL conditions for $(s_1,t_1)$-insdel and $(s_2,t_2)$-insdel errors
  are equivalent, 
  where $s_1+t_1=s_2+t_2$,
  i.e.,
  for any quantum code $\QU{C} \subset \C^{l \otimes n}$,
  $\QU{C}$ corrects $(s_1,t_1)$-insdel errors
  iff
  $\QU{C}$ corrects $(s_2,t_2)$-insdel errors.
\end{theo}}

\Corr{If 
$\QU{C}$ satisfies the KL condition for $(s_1,t_1)$-insdel errors,
$\QU{C}$ also satisfies it for $(s_2,t_2)$-insdel errors,
where $s_1 \ge s_2$ and $t_1 \ge t_2$
\cite[The proof of Lemma V.$2$]{shibayama2021equivalence}.
Hence,
Theorem~\ref{theo:shibayama2021equivalence} is extended as follows:}

\Corr{\begin{coro}\cite{shibayama2021equivalence}
  \label{coro:shibayama2021equivalence}
  The KL condition for $t$ deletion errors is equivalent to 
  the KL conditions for $(s',t')$-insdel errors
  for any $s',t' \in \N$ satisfying $s'+t' \le t$,
  i.e.,
  for any quantum code $\QU{C} \subset \C^{l \otimes n}$,
  $\QU{C}$ corrects $t$ deletion errors
  iff
  $\QU{C}$ corrects a total of at most $t$ deletion and separable insertion errors.
\end{coro}}

\section{Separability of Restricted Insertion Errors}
\label{sec:ins_only_BSins}
\Corr{In this section,
we show that insertion errors are limited to BS insertions.}

\Corr{\begin{theo}
  \label{theo:ins_is_BSins}
  Define $A$, $B$, $\QU{H}_A$, $\QU{H}_B$ and $\text{Tr}_B$
  as in Theorem~\ref{theo:after_ins_gene}.
Let $\text{id}_A$ denote the identity map on $\QU{L}(\QU{H}_A)$.
  For a quantum channel $\QU{E} : \QU{L}(\QU{H}_A) \to \QU{L}(\QU{H}_A \otimes \QU{H}_B)$  
  such that $\text{Tr}_B \circ \QU{E} = \text{id}_A$,
  there exists $\pi_B \in S(\QU{H}_B)$ such that
  \begin{align*}
    \forall \rho \in S(\QU{H}_A) ~~~
    \QU{E}(\rho) = \rho \otimes \pi_B
  \end{align*}
  In particular,
  for a set $Q = \{ q_1, q_2, ..., q_t \} \in \subint{n+t}{t}$ and
  a quantum channel $\QU{I}_Q : \QU{L}(\C^{l\otimes n}) \to \QU{L}(\C^{l\otimes (n+t)})$
  such that
  \begin{align*}
    \forall \rho \in S(\C^{l \otimes n}) ~~~
    \QU{D}_Q \circ \QU{I}_Q( \rho )
    =
    \rho,
  \end{align*}
  there exists
  $\pi_Q \in S(\C^{l\otimes t})$
  such that
  \begin{align}
    \forall \rho \in S(\C^{l \otimes n}) ~~~
    \QU{I}_Q(\rho)
    =
    \mat{U}_{\tau^Q}
    \left(
      \rho\otimes\pi_Q
    \right)
    \mat{U}_{\tau^Q}^{\dagger},
    \label{eq:channel_product_general}
  \end{align}
  where
  $\mat{U}_{\tau^Q}$ is defined as in Theorem~\ref{theo:after_ins_gene_ver_shiba}.
\end{theo}}

\Corr{\begin{IEEEproof}
  We denote by $R$
  a reference system
  whose Hilbert space $\QU{H}_R$ is isomorphic to $\QU{H}_A$.
  Let
  $\{\ket{i}_A\}_{i=1}^d$ and $\{\ket{i}_R\}_{i=1}^d$
  be orthonormal bases of $\QU{H}_A$ and $\QU{H}_R$, 
  respectively,
  where $d$ is the dimension of $\QU{H}_A$ and $\QU{H}_R$.
  For simplicity,
  for a system $S$,
  we denote $\ket{ \cdot }_S \bra{ \cdot }_S$ by $\ket{\cdot}\bra{\cdot}_S$.
  In addition,
  let $\text{id}_R$ denote the identity map on $\QU{L}(\QU{H}_R)$.
  Define
  \begin{align}
    \ket{\Phi}_{RA}
    :=
    \frac{1}{\sqrt{d}}
    \sum_{i=1}^d
    \ket{i}_R \otimes \ket{i}_A.
    \label{eq:def_ket_RA}
  \end{align}
  There is a one-to-one correspondence between
  linear maps and their Choi representations
  \cite{watrous2018}.
  Therefore,
  we focus on the Choi state corresponding to $\QU{E}$,
  \begin{align}
    J_{\QU{E}}
    :=
    (\text{id}_R \otimes \QU{E})
    (\ket{\Phi}\bra{\Phi}_{RA}).
    \label{eq:def_Je}
  \end{align}
  We obtain
  \begin{align*}
    \text{Tr}_B( J_{\QU{E}} )
    &=
    (\text{id}_R \otimes (\text{Tr}_B \circ \QU{E}) )( \ket{\Phi}\bra{\Phi}_{RA} )\\
    &=
    \ket{\Phi}\bra{\Phi}_{RA},
  \end{align*}
  i.e.,
  the subsystem of $J_{\QU{E}}$ is pure.
  Hence,
  from Theorem~\ref{theo:after_ins_gene},
  $J_{\QU{E}}$ is represented by
  \begin{align*}
    J_{\QU{E}}
    &= 
    \ket{\Phi}\bra{\Phi}_{RA} \otimes \pi_B
    \nonumber\\
    &=
    \frac{1}{d}
    \sum_{i,j=1}^d
    \ket{i}\bra{j}_R \otimes \ket{i}\bra{j}_A        
    \otimes \pi_B
  \end{align*}
  for some $\pi_B \in S(\QU{H}_B)$.
  Denote the identity matrix on $\QU{H}_A \otimes \QU{H}_B$ by $\mat{I}_{AB}$.
  Then,
  it follows that 
  \begin{align}
    (\bra{i'}_R \otimes \mat{I}_{AB}) J_{\QU{E}}
    (\ket{j'}_R \otimes \mat{I}_{AB})
    &=    
    \frac{1}{d}
    \ket{i'}\bra{j'}_A \otimes \pi_B
    \label{eq:Je_denote}
  \end{align}
  for any $i' ,j' \in \setint{d}$.
  By combining Eqs.~\eqref{eq:def_ket_RA} and \eqref{eq:def_Je},
  we get
  \begin{align*}
    J_{\QU{E}}
    &=
    (\text{id}_R \otimes \QU{E})
    \left(
    \frac{1}{d}
    \sum_{i,j=1}^d
    \ket{i}\bra{j}_R \otimes \ket{i}\bra{j}_A
    \right)\\
    &=
    \frac{1}{d}
    \sum_{i,j=1}^d
    \ket{i}\bra{j}_R \otimes \QU{E}(\ket{i}\bra{j}_A).
  \end{align*}
  It follows that 
  \begin{align}
    (\bra{i'}_R \otimes \mat{I}_{AB}) J_{\QU{E}}
    (\ket{j'}_R \otimes \mat{I}_{AB})
    &=
    \frac{1}{d}
    \QU{E}(\ket{i'}\bra{j'}_A).
    \label{eq:J_E_kankei}
  \end{align}
  By combining Eqs.~\eqref{eq:Je_denote} and \eqref{eq:J_E_kankei},
  we have
  $\QU{E}(\ket{i'}\bra{j'}_A) = \ket{i'}\bra{j'}_A \otimes \pi_B$
  for any $i' ,j' \in \setint{d}$.
  Therefore,
  if $\rho = \sum_{i,j = 1}^d c_{i,j} \ket{i}\bra{j}_A$ for $c_{i,j} \in \C$,  
  \begin{align*}
    \QU{E}(\rho)
    &=
    \sum_{i,j = 1}^d c_{i,j} \QU{E}( \ket{i}\bra{j}_A )\\
    &=
    \sum_{i,j = 1}^d c_{i,j} \ket{i}\bra{j}_A \otimes \pi_B
    =
    \rho \otimes \pi_B.
  \end{align*}
  In particular, 
  by regarding the $t$ qudits inserted at $Q$ as the subsystem $B$ and 
  applying $\mat{U}_{\tau^Q}$ to $\rho \otimes \pi_B$,
  an insertion error $\QU{I}_Q$ is represented by Eq.~\eqref{eq:channel_product_general}.
\end{IEEEproof}}

\Corr{For a unit vector
$\ket{\phi_r} \in \C^{l \otimes t}$,
define
$\mat{I}^n_{Q,\ket{\phi_r}} := \mat{U}_{\tau^Q} ( \mat{I}_{l^n} \otimes \ket{\phi_r} )$.
Then,
for the spectral decomposition $\pi_B = \sum_r \lambda_r \ket{\phi_r}\bra{\phi_r}$,
the Kraus operator for the BS insertion $\rho \mapsto \rho \otimes \pi_B$
is defined as
\begin{align}
  \mat{J}^n_{Q,r}
  :=
  \sqrt{\lambda_r} \mat{I}^n_{Q,\ket{\phi_r}},
  \label{eq:BS-ins_Kra}
\end{align}
similarly to the Kraus operator for the separable insertion.
Note that
if $\ket{\phi_r}$ is restricted to a product state,
this Kraus operator equals that for a separable insertion.}

\Corr{\begin{rem}
  \label{rem:commutative}
  Recall the notation in Section~\ref{subsec:def_insdel}.
  When $t_1$ deletion errors occur after $t_2$ BS insertions,
  the state changes as
  \begin{align}
    \rho
    &\mapsto
    \rho \otimes \pi_B
    \mapsto
    \QU{D}_{P_1}(\rho) \otimes \QU{D}_{P_2}(\pi_B),
    \label{eq:rem_comm_exam}
  \end{align}
  where $t_1 = |P_1| + |P_2|$.
  The state $\QU{D}_{P_1}(\rho) \otimes \QU{D}_{P_2}(\pi_B)$ can be regarded as 
  the state obtained by a $(t_2 - |P_2|)$ BS insertion after a $(t_1 - |P_2|)$ deletion error.
  By repeating a swap of deletion and BS insertion errors,
  the Kraus operator for a composite error can be represented in the form
  \begin{align*}
    \mat{J}^{\,n-t_2+k}_{Q,r}
    \mat{D}^{\,n-t_2+k}_{P,\bra{\bm{\Psi}}},
  \end{align*}
  where 
  $k$ denotes the number of inserted qudits that are deleted,
  $|P|=t_2-k$, and $|Q|=t_1-k$. 
  In other words,
  for a composite error consisting of deletion and BS insertion errors,
  it can be represented as 
  a composite error in which the BS insertion errors occur after the deletion errors.
  For example,
  the Kraus operator of the quantum channel of Eq.~\eqref{eq:rem_comm_exam} 
  is represented by the form 
  $\mat{J}^{\,n-t_2+|P_2|}_{Q,r}\mat{D}^{\,n-t_2+|P_2|}_{P,\bra{\bm{\Psi}}}$.
\end{rem}}

\section{Relationship between the Error-Correcting Capabilities for Deletion and Restricted Insertion Errors}
\label{sec:KL_BSins=del}
\Corr{In this section,
we show that
the KL condition for $t$ separable insertion errors
is equivalent to
the KL condition for BS insertion errors.
By combining this result, Corollary~\ref{coro:shibayama2021equivalence}, and Theorem~\ref{theo:ins_is_BSins},
we derive that
a quantum code $\QU{C}$ corrects $t$ deletion errors
iff
$\QU{C}$ corrects 
a total of at most $t$ insertion and deletion errors.}

\Corr{\begin{theo}
  \label{theo:KL_SI=KL_BSI}
  The KL condition for $t$ separable insertion errors
  is equivalent to
  the KL condition for BS insertion errors.
  In other words,
  a quantum code $\QU{C}$ corrects $t$ separable insertion errors
  iff
  $\QU{C}$ corrects $t$ BS insertion errors.
\end{theo}}

\begin{IEEEproof}
  \Corr{For the set $S$ of matrices on $\QU{L}(\QU{H})$,
  let $\text{span}~S$ be the set of complex linear combinations of the elements in $S$.
  Then,
  it is known that
  if a Kraus set $\{ \mat{E}_\alpha \}_\alpha$ satisfies the KL condition,
  a Kraus set $\{ \mat{F}_\alpha \}_\alpha \subset \text{span}~\{ \mat{E}_\alpha \}_\alpha$ satisfies the KL condition
  \cite{nielsen2016quantum}.
  Hence,
  if $\text{span}~\{ \mat{E}_\alpha \}_\alpha = \text{span}~\{ \mat{F}_\alpha \}_\alpha$ holds,
  the KL condition for $\{E_\alpha\}_\alpha$ holds
  is equivalent to
  that for $\{F_\alpha\}_\alpha$ holds.}

  \Corr{We now fix $Q$
  and show that 
  the spaces spanned by $\mat{J}^n_{Q,r}$ and $\mat{I}^n_{Q,\ket{\bm{\Psi}}}$
  are equal.
  The operator $\mat{I}^n_{Q,\ket{\bm{\Psi}}}$
  coincides with $\mat{J}^n_{Q,r}$
  when
  $\pi_B = \ket{\bm{\Psi}}\bra{\bm{\Psi}}$.
  Conversely,
  $\mat{J}^n_{Q,r}$ is represented as
  a complex linear combination of $\mat{I}^n_{Q,\ket{\bm{\Psi}}}$.
  Indeed,
  from Eqs.~\eqref{eq:S_ins_Kraus} and \eqref{eq:BS-ins_Kra},
  by expanding $\ket{\phi_r} = \sum_{\bm{x} \in \Z_l^t} c_{r,\bm{x}} \ket{\bm{x}}$,
  \begin{align*}
    \mat{J}^n_{Q,r}
    =
    \sqrt{\lambda_r}
    \sum_{\bm{x} \in \Z_l^t}
    c_{r,\bm{x}}
    \mat{I}^n_{Q,\ket{\bm{x}}}    
  \end{align*}
  holds.
  Thus,
  the space spanned by
  $\mat{J}^n_{Q,r}$
  is equal to
  the space spanned by
  $\mat{I}^n_{Q,\ket{\bm{\Psi}}}$.
  Since this holds for any insertion position $Q$,
  the KL condition for $t$ separable insertion errors
  is equivalent to
  the KL condition for BS insertion errors.}
\end{IEEEproof}

\Corr{\begin{theo}
  \label{theo:tdel_iff_tindel}
  The KL condition for $t$ deletion errors
  is equivalent to
  the KL condition for a total of at most $t$ BS insertion and deletion errors.
  In other words,
  for a quantum code $\QU{C}$,
  $\QU{C}$ corrects $t$ deletion errors
  iff
  $\QU{C}$ corrects a total of at most $t$ deletion and BS insertion errors.
\end{theo}}

\begin{IEEEproof}
  \Corr{Consider an arbitrary composite error consisting of
  $t_1$ BS insertion errors and $t_2$ deletion errors,
  where $t_1+t_2 \leq t$.  
  From Remark~\ref{rem:commutative},
  its Kraus operator can be represented as $\mat{J}^{\,n-t_2+k}_{Q,r} \mat{D}^{\,n-t_2+k}_{P,\bra{\bm{\Psi}}}$,
  where
  $k$ is the number of inserted qudits that are deleted,
  $|P|=t_2-k$, and $|Q|=t_1-k$.}
  
  \Corr{As in the proof of Theorem~\ref{theo:KL_SI=KL_BSI},
  the spaces spanned by
  $\mat{I}^{\,n-t_2+k}_{Q,\ket{\bm{\Psi}}} \mat{D}^{\,n-t_2+k}_{P,\bra{\bm{\Psi}}}$ and 
  $\mat{J}^{\,n-t_2+k}_{Q,r} \mat{D}^{\,n-t_2+k}_{P,\bra{\bm{\Psi}}}$
  are equal.
  Thus,
  the KL condition for $(t_2-k)$ deletion and $(t_1-k)$ separable insertion errors 
  is equivalent to
  that for $(t_2-k)$ deletion and $(t_1-k)$ BS insertion errors holds.
  Here,
  from Corollary~\ref{coro:shibayama2021equivalence},
  if a quantum code $\QU{C}$ satisfies the KL condition for $t$ deletion errors,
  $\QU{C}$ satisfies the KL condition for $(t_2-k)$ deletion and $(t_1-k)$ separable insertion errors
  since $|P| + |Q| \leq t$.
  By the equivalence shown above,
  $\QU{C}$ also satisfies the KL condition for
  $(t_2-k)$ deletion errors and $(t_1-k)$ BS insertion errors. 
  Since the original composite error is arbitrary,
  if $\QU{C}$ satisfies the KL condition for $t$ deletion errors,
  $\QU{C}$ satisfies the KL condition for a total of at most $t$ deletion and BS insertion errors.}
  
  \Corr{Conversely,
  $t$ deletion errors are a special case of
  a total of $t$ BS insertion and deletion errors,
  obtained by taking $t_1=0$ and $t_2=t$.
  Hence, the converse also holds.}
\end{IEEEproof}

\Corr{\begin{rem}
  The KL condition guarantees the existence of a recovery operator, 
  even for non-square Kraus operators \cite{shibayama2025necessary}.
  It is claimed in \cite{hagiwara2024pure} that 
  the following code cannot correct a single deletion error but satisfies the KL condition.
  \begin{align*}
    \{
      \alpha \ket{00} + \beta \ket{11}
      \mid
      \alpha, \beta \in \C,
      |\alpha|^2 + |\beta|^2 = 1
    \}
  \end{align*}
  However,
  this code does not satisfy the KL condition.
  Indeed,
  for any $p \in \setint{2}$,
  since $\mat{D}^2_{p, \bra{0}} \ket{00} = \ket{0}$ and $\mat{D}^2_{p, \bra{0}} \ket{11} = 0 \cdot \ket{1}$,
  \begin{align*}
    \bra{00} \mat{D}^{2~\dagger}_{p, \bra{0}} \mat{D}^2_{p, \bra{0}} \ket{00}
    =
    1
    \neq
    0
    =
    \bra{11} \mat{D}^{2~\dagger}_{p, \bra{0}} \mat{D}^2_{p, \bra{0}} \ket{11}    
  \end{align*}
  holds.
\end{rem}}

\Corr{We obtain the following
by combining
Corollary~\ref{coro:shibayama2021equivalence}
and
Theorems~\ref{theo:ins_is_BSins},~\ref{theo:tdel_iff_tindel}.}

\Corr{\begin{coro}
  \label{coro:tdel_iff_tindel}
  Satisfying the KL condition for $t$ deletion errors
  is equivalent to
  satisfying the KL condition for
  $s'$ deletion errors and $t'$ insertion errors
  for any $s',t' \in \N$ satisfying
  $s'+t' \le t$.
  In other words,
  for any quantum code $\QU{C} \subset \C^{l \otimes n}$,
  $\QU{C}$ corrects $t$ deletion errors
  iff
  $\QU{C}$ corrects a total of at most $t$ deletion and insertion errors.
\end{coro}}

\section{Generalized Insertion Error}
\label{sec:gene_ins_prop}
\Corr{
Although an overall transformation can be represented by a quantum channel, 
when we regard the transformation as a composite error consisting of deletion and insertion errors, 
the insertion part does not satisfy the quantum channel requirement.
We call such insertion {\it generalized insertion} error.
We explain the need for generalized insertion errors
by considering a permutation of qudits \cite{hagiwara2025introduction}.
For example,
consider the following permutation:
\begin{align*}
  &\ket{\phi_1}
  :=
  \tfrac{1}{\sqrt{2}}( \ket{001}+\ket{110} ) \\
  &\mapsto 
  \ket{\phi_2}
  :=
  \tfrac{1}{\sqrt{2}}( \ket{100}+\ket{011} ).
\end{align*}
If we denote the three qudits of $\ket{\phi_1}$ by $q_1, q_2, q_3$,
the transformation can be regarded as
a composite error consisting of the deletion of $q_3$ and its insertion at the first position.
On the other hand,
the insertion part cannot be represented by
an insertion error with the quantum channel requirement 
since
$\ket{\phi_2}$ is not a product state.
In fact,
a transformation that maps arbitrary input states
to their purifications cannot be represented by a quantum channel \cite{kleinmann2006physical}.
In general,
although an overall transformation can be represented by a quantum channel, 
when we regard the transformation as a composite error consisting of deletion and insertion errors, 
the insertion part does not satisfy the quantum channel requirement.
To describe such insertion parts without the quantum channel requirement,
we formally define generalized insertion errors $\QU{I}_Q^{\text{G}}$ as follows:}

\Corr{\begin{defi}
  \label{def:gene_ins}
  For $\rho\in\state{l}{n}$ and
  $Q\in\subint{n+t}{t}$,
  define
  \begin{align*}
    \QU{I}_Q^{\mathrm{G}}(\rho)
    :=
    \left\{
      \sigma\in\state{l}{n+t}
      \mid
      \QU{D}_Q(\sigma)=\rho
    \right\}.
  \end{align*}
  An error that transforms $\rho$ into
  a state $\sigma \in \QU{I}_Q^{\text{G}}(\rho)$
  is called a {\it generalized insertion} error
  at the positions $Q$.
\end{defi}}

\Corr{\begin{rem}
  Since $\QU{I}_Q (\rho) \in \QU{I}_Q^{\mathrm{G}}(\rho)$,
  a generalized insertion error is 
  a generalization of an insertion error.
\end{rem}}

\Corr{\begin{rem}
  The transformation 
  \begin{align*}
    &\rho_1 = \tfrac{1}{2}(\ket{0}\bra{0} + \ket{1}\bra{1}) \\
    &\mapsto 
    \rho_2 = \frac{1}{2} (\ket{00}+\ket{11})(\bra{00}+\bra{11}) 
  \end{align*}
  is not an insertion error,
  but is a generalized insertion error
  since
  $\rho_1 = \text{Tr}_1(\rho_2)$
  and $\rho_2$ is not a product state.  
  In general,
  even if $\rho_1$ is obtained from $\rho_2$
  by $t$ deletion errors,
  $\rho_2$ is not necessarily obtained from $\rho_1$
  by $t$ insertion errors.
  On the other hand,
  $\rho_2$ is obtained from $\rho_1$ by a generalized insertion error 
  (see Section~\ref{subsec:prop_indel}).
  This represents one of the major differences
  between insertion errors and generalized insertion errors
  and
  leads to differences in the error-correcting capabilities of codes.
  Specifically,
  for a quantum code $\QU{C}$,
  $\QU{C}$ corrects a total of $t$ deletion and insertion errors
  iff
  $\QU{C}$ corrects $t$ deletion errors 
  by Corollary~\ref{coro:tdel_iff_tindel},
  and
  $\QU{C}$ corrects a total of $t$ deletion and generalized insertion errors 
  if
  $\QU{C}$ corrects $t$ deletion errors (see Theorem~\ref{theo:1del_iff_1ins_2}).
\end{rem}}

\subsection{State after Generalized Insertion Error}
\Corr{This section shows that
$\QU{I}_Q^{\text{G}}(\rho)$ is expressed as follows:}

\begin{coro}
  \label{coro:insset}
  Define $\mat{U}_{\tau^Q}$ as in Theorem~\ref{theo:after_ins_gene_ver_shiba}
  and
  $\rho \in \state{l}{n}$ as in Eq.~\eqref{eq:spec}.
  Then,
  \begin{align*}
    \QU{I}_Q^{\text{G}}(\rho)
    =
    \Bigg\{
    &\mat{U}_{\tau^Q}
    \bigg(
    \sum_{\bm{x},\bm{y} \in \Z_{l}^{n}}
    \sqrt{ p_{\bm{x}} p_{\bm{y}} }
    \ket{\bm{x}_L}\bra{\bm{y}_L} \otimes \mat{A}_{\bm{x},\bm{y}}
    \bigg)
    \mat{U}_{\tau^Q}^\dagger
    \ge 0\\
    &\Big|
    \mat{A}_{\bm{x},\bm{y}} \in M(l^t),
    \text{tr}( \mat{A}_{\bm{x},\bm{y}} ) = 0 ~ (\bm{x} \neq \bm{y}),\\
    &~\mat{A}_{\bm{x},\bm{y}}^\dagger = \mat{A}_{\bm{y},\bm{x}},
    \mat{A}_{\bm{x},\bm{x}} \in \state{l}{t}
    \Bigg\}.
  \end{align*}
\end{coro}

A $k \times k$ principal submatrix of $\mat{M} \in M(n)$
is
a submatrix that lies on the same set of $k$ rows and columns,
and
a $k \times k$ principal minor is the determinant of a $k \times k$
principal submatrix.
\begin{exam}
  The $k \times k$ principal minors of
  $\begin{pmatrix} 1 & 2 & 3\\ 4 & 5 & 6\\ 7 & 8 & 9\end{pmatrix}$ are
    \begin{align*}
      &\begin{vmatrix}
        1 
      \end{vmatrix},
      \begin{vmatrix}
        5
      \end{vmatrix},
      \begin{vmatrix}
        9
      \end{vmatrix},
      &&(k=1)\\
      &\begin{vmatrix}
        1 & 2\\
        4 & 5
      \end{vmatrix},
      \begin{vmatrix}
        1 & 3\\
        7 & 9
      \end{vmatrix},
      \begin{vmatrix}
        5 & 6\\
        8 & 9
      \end{vmatrix},
      &&(k=2)\\
      &\begin{vmatrix}
        1 & 2 & 3\\
        4 & 5 & 6\\
        7 & 8 & 9
      \end{vmatrix}.
      &&(k=3)
    \end{align*}
\end{exam}

The following are  well-known results for positive semidefinite matrices
(e.g., see \cite{horn1985matrix}).

\begin{lem}
  \label{lem:hansei}
  For a Hermitian matrix $\mat{M}$,
  $\mat{M} \ge 0$
  if and only if
  all principal minors of $\mat{M}$ are nonnegative.
\end{lem}

\begin{lem}
  \label{lem:sub_seitei_kai}
  For any matrix $\mat{A}$,
  if $\mat{M} \ge 0$,
  then
  $\mat{A} \mat{M} \mat{A}^\dagger \ge 0$.
\end{lem}

\begin{lem}
  \label{lem:px=0->cxy=0}
  For a positive semidefinite matrix $\mat{M} = (m_{i,j}) \in M(n)$,
  if $m_{i,i} = 0$,
  \begin{align*}
    \forall j \in \setint{n}
    ~~~
    m_{i,j} = m_{j,i} = 0
  \end{align*}
  holds.
\end{lem}

The following theorem \Corr{characterizes}
the quantum states after $\QU{I}^{\text{G}}_{\setint{n+1,n+t}}$.
\begin{theo}
  \label{theo:ins_state}
  Suppose
  $\rho \in \state{l}{n}$ is given by Eq.~\eqref{eq:spec}.
  Then,
  for $t \in \N$,
  $\sigma \in \QU{I}^{\text{G}}_{\setint{n+1,n+t}}(\rho)$ is represented by
  \begin{align}
    \sigma
    =&
    \sum_{\bm{x},\bm{y} \in \Z_{l}^{n}}
    \sqrt{ p_{\bm{x}} p_{\bm{y}} }
    \ket{\bm{x}_L}\bra{\bm{y}_L} \otimes \mat{A}_{\bm{x},\bm{y}},
  \end{align}
  where
  $\mat{A}_{\bm{x},\bm{x}} \in \state{l}{t}$ and
  $\mat{A}_{\bm{x},\bm{y}} \in M(l^t)$ satisfies
  $\mat{A}_{\bm{x},\bm{y}}^\dagger = \mat{A}_{\bm{y},\bm{x}}$ and
  $\text{tr}( \mat{A}_{\bm{x},\bm{y}} ) = 0$ ($\bm{x} \neq \bm{y}$).
\end{theo}
\begin{IEEEproof}
  Since $\{ \ket{\bm{x}_L} \}_{ \bm{x} \in \Z_{l}^{n} }$ is   
  an orthonormal basis for $\pure{l}{n}$,
  $\{ \ket{\bm{x}_L} \otimes \ket{ \bm{z} } \}_{\bm{x} \in \Z_{l}^{n}, \bm{z} \in \Z_l^t}$
  is
  an orthonormal basis for $\pure{l}{(n+t)}$.
  From Eq.~\eqref{eq:orth_basis_xL},
  $\sigma \in M(l^{(n+t)})$ is rewritten as
  \begin{align*}
    \sigma
    &=
    \sum_{\bm{x},\bm{y} \in \Z_{l}^{n}}
    \sum_{\bm{z}_1,\bm{z}_2 \in \Z_l^t}
    c_{\bm{x},\bm{z}_1,\bm{y},\bm{z}_2}
    (\ket{\bm{x}_L} \otimes \ket{ \bm{z}_1 })
    (\bra{\bm{y}_L} \otimes \bra{ \bm{z}_2 })\\
    &=
    \sum_{\bm{x},\bm{y} \in \Z_{l}^{n}}
    \ket{\bm{x}_L}\bra{\bm{y}_L} \otimes
    \mat{B}_{\bm{x},\bm{y}},
  \end{align*}
  where
  $\mat{B}_{\bm{x},\bm{y}} := \sum_{\bm{z}_1,\bm{z}_2 \in \Z_l^t} c_{\bm{x},\bm{z}_1,\bm{y},\bm{z}_2} \ket{ \bm{z}_1 }\bra{ \bm{z}_2 } \in M(l^t)$.
  Now,
  we will determine the conditions on $\mat{B}_{\bm{x},\bm{y}}$ from
  \begin{align*}
    \sigma
    =
    \sum_{\bm{x},\bm{y} \in \Z_{l}^{n}}
    \ket{\bm{x}_L}\bra{\bm{y}_L} \otimes \mat{B}_{\bm{x},\bm{y}}
    \in
  \QU{I}^{\text{G}}_{\setint{n+1,n+t}}(\rho).
  \end{align*}
  Since $\QU{D}_{\setint{n+1,n+t}}(\sigma) = \rho$,
  we get
  \begin{align}
    \text{tr} ( \mat{B}_{\bm{x},\bm{y}} )
    &=
    \begin{cases}
      p_{ \bm{x} } &( \text{if~} \bm{x} = \bm{y} )\\
      0 &( \text{if~} \bm{x} \neq \bm{y} )
    \end{cases}
    \label{eq:del_after}
  \end{align}
  Note that
  $\mat{U}' := \sum_{ \bm{x} \in \Z_l^n } \ket{\bm{x}} \bra{\bm{x}_L} \otimes \mat{I}_{l^t}$
  is unitary.
  
  From Lemma \ref{lem:sub_seitei_kai},
  $\sigma \ge 0$
  if and only if
  \begin{align*}
    \mat{U}' \sigma \mat{U}'^\dagger 
    =
    \sum_{\bm{x},\bm{y} \in \Z_{l}^{n}}
    \ket{\bm{x}}\bra{\bm{y}} \otimes
    \mat{B}_{\bm{x},\bm{y}}
    \ge 0.
  \end{align*}
  From Lemma \ref{lem:hansei},
  $\mat{B}_{\bm{x},\bm{x}} \ge 0$ and
  all diagonal elements of $\mat{B}_{\bm{x},\bm{x}}$ are nonnegative.
  From Eq.~\eqref{eq:del_after},
  all diagonal elements of $\mat{B}_{\bm{x},\bm{x}}$ are zero
  if $p_{\bm{x}} = 0$.
  Combining this with Lemma \ref{lem:px=0->cxy=0},
  if $p_{ \bm{x} } = 0$,
  $\mat{B}_{\bm{x},\bm{y}}$ and $\mat{B}_{\bm{y},\bm{x}}$ are equal to
  the zero matrix
  for any $\bm{y} \in \Z_l^n$.
  Hence,
  \begin{align*}
    \sigma
    &=
    \sum_{\bm{x},\bm{y} \in \Z_{l}^{n}}
    \sqrt{ p_{\bm{x}} p_{\bm{y}} }
    \ket{\bm{x}_L}\bra{\bm{y}_L} \otimes \mat{A}_{\bm{x},\bm{y}},
  \end{align*}
  where
  $\mat{A}_{\bm{x},\bm{y}} := \mat{B}_{\bm{x},\bm{y}} / \sqrt{ p_{\bm{x}} p_{\bm{y}} } \in M(l^t)$
  if
  $p_{\bm{x}} \neq 0$ and $p_{\bm{y}} \neq 0$.
  Since $\mat{B}_{\bm{x},\bm{x}} \ge 0$ and $\sigma$ is Hermitian,
  we get
  $\mat{A}_{\bm{x},\bm{x}} \ge 0$ and
  $\mat{A}_{\bm{x},\bm{y}}^\dagger = \mat{A}_{\bm{y},\bm{x}}$. 
  In addition,
  Eq.~\eqref{eq:del_after} yields
  \begin{align*}    
    \text{tr} ( \mat{A}_{\bm{x},\bm{y}} )
    &=
    \frac{1}{ \sqrt{ p_{\bm{x}} p_{\bm{y}} } }
    \text{tr}
    \left( \mat{B}_{\bm{x},\bm{y}} \right)
    =
    \begin{cases}
      1 & (\text{if~} \bm{x} = \bm{y})\\
      0 & (\text{if~} \bm{x} \neq \bm{y})
    \end{cases}.
  \end{align*}
  From the above,
  we obtain $\mat{A}_{\bm{x},\bm{x}} \in \state{l}{t}$.
\end{IEEEproof}

By applying the index permutation to Theorem \ref{theo:ins_state},
generalized insertion errors at arbitrary insertion positions can be expressed as follows:
\begin{coro}
  \label{coro:ins_state}
  Define $\mat{U}_{\tau^Q}$ as in Theorem~\ref{theo:after_ins_gene_ver_shiba}
  and
  $\mat{A}_{\bm{x},\bm{y}} \in M(l^t)$ as in Theorem \ref{theo:ins_state}.
  Then,
  for a quantum state $\rho \in \state{l}{n}$ in Eq.~\eqref{eq:spec},
  $\sigma \in \QU{I}_Q^{\text{G}}(\rho)$ is represented by
  \begin{align}
    \sigma
    =
    \mat{U}_{\tau^Q}
    &
    \bigg(
    \sum_{\bm{x},\bm{y} \in \Z_{l}^{n}}
    \sqrt{ p_{\bm{x}} p_{\bm{y}} }
    \ket{\bm{x}_L}\bra{\bm{y}_L} \otimes \mat{A}_{\bm{x},\bm{y}}
    \bigg)
    \mat{U}_{\tau^Q}^\dagger.
    \label{eq:theo:insstate_con_Q}
  \end{align}
\end{coro}

\Corr{Corollary~\ref{coro:ins_state} characterizes the states after generalized insertion errors
at arbitrary insertion positions.
In addition,
Corollary~\ref{coro:ins_state} is used
to analyze composite deletion and generalized insertion errors
in Example~\ref{exam:swap_TinsTdel}.}

\begin{rem}
  If $\rho$ is pure,
  we get
  Eq.\eqref{eq:pure_ins} in Theorem \ref{theo:after_ins_gene_ver_shiba}
  from Eq.\eqref{eq:theo:insstate_con_Q} in Corollary \ref{coro:ins_state}.
\end{rem}

Noting that
Eq.~\eqref{eq:theo:insstate_con_Q} does not guarantee $\sigma \ge 0$,
we get \Corr{Corollary \ref{coro:insset}} from Corollary \ref{coro:ins_state}.

\begin{exam}
  \label{exam:ins_set}
  For $p_0, p_1 \in \NNR$,
  define
  \begin{align*}
    \rho
    :=
    p_0\ket{00}\bra{00}
    +
    p_1\ket{11}\bra{11}
    \in
    \state{2}{2}.
  \end{align*}
  Then,
  from Corollary \ref{coro:ins_state},
  $\sigma_i \in \QU{I}^{\text{G}}_{\{i\}}(\rho)$ with $i=1,2,3$ are
  \begin{align*}
    \sigma_1
    =&
    p_0 ~ \pi_{00} \otimes \ket{00}\bra{00}
    +
    p_1 ~ \pi_{11} \otimes \ket{11}\bra{11}\\
    &+
    \sqrt{p_1 p_0} \mat{A} \otimes \ket{11}\bra{00}
    +
    \sqrt{p_0 p_1} \mat{A}^\dagger \otimes \ket{00}\bra{11},\\
    \sigma_2
    =&
    p_0 \ket{0}\bra{0} \otimes \pi_{00} \otimes \ket{0}\bra{0}
    +
    p_1 \ket{1}\bra{1} \otimes \pi_{11} \otimes \ket{1}\bra{1}\\
    &+
    \sqrt{p_1 p_0} \ket{1}\bra{0} \otimes \mat{A} \otimes \ket{1}\bra{0}\\
    &+
    \sqrt{p_0 p_1} \ket{0}\bra{1} \otimes \mat{A}^\dagger \otimes \ket{0}\bra{1},\\
    \sigma_3
    =&
    p_0 \ket{00}\bra{00} \otimes \pi_{00}
    +
    p_1 \ket{11}\bra{11} \otimes \pi_{11}\\
    &+
    \sqrt{p_1 p_0} \ket{11}\bra{00} \otimes \mat{A}
    +
    \sqrt{p_0 p_1} \ket{00}\bra{11} \otimes \mat{A}^\dagger,
  \end{align*}
  where
  $\pi_{00},\pi_{11} \in S(\C^2)$,
  $\mat{A} \in M(2)$ such that $\text{tr}(\mat{A}) = 0$,
  and $\sigma_i \ge 0$.
\end{exam}

\subsection{State after Deletion and Generalized Insertion Error}
\subsubsection{Notations}
Let $\QU{D}^s(\rho)$ be 
the set of states resulting from $s$ deletion errors $\QU{D}^s$ applied to $\rho$
and
$\QU{I}^{\text{G},t}(\rho)$ be 
the set of states resulting from $t$ generalized insertion errors $\QU{I}^{\text{G},t}$ applied to $\rho$,
i.e.,
for $\rho \in \state{l}{n}$,
\begin{align*}
  \QU{D}^s(\rho)
  &=
  \bigcup_{P \in \subint{n}{s}}
  \{
  \QU{D}_P(\rho)
  \},\\
  \QU{I}^{\text{G},t}(\rho)
  &=  
  \bigcup_{Q \in \subint{n+t}{t}}
  \QU{I}^{\text{G}}_Q(\rho).
\end{align*}
Define
$\QU{D}^0(\rho) = \QU{I}^{\text{G},0}(\rho) := \{ \rho \}$.

For two errors $\QU{E}_1$ and $\QU{E}_2$,
define
$\QU{E}_2\circ\QU{E}_1(\rho) := \bigcup_{ \sigma \in \QU{E}_1(\rho) } \QU{E}_2(\sigma)$.
If a state-valued error $\QU{E}$ is used in a set-valued composition,
we identify $\QU{E}(\rho)$ with the singleton $\{\QU{E}(\rho)\}$.
For example,
for $\rho \in \state{l}{n}$,
\begin{align*}
  \QU{D}^s \circ \QU{I}^{\text{G},t}( \rho )
  &=
  \bigcup_{\sigma \in \QU{I}^{\text{G},t} (\rho)}
  \QU{D}^s (\sigma),\\
  \QU{I}^{\text{G},t} \circ \QU{D}^s( \rho )
  &=
  \bigcup_{Q \in \subint{n-s+t}{t}}
  \bigcup_{\sigma \in \QU{D}^s (\rho)}
  \QU{I}_Q^{\text{G}} (\sigma).
\end{align*}
Moreover,
for $\QU{X} \subset \state{l}{\N}$,
define
\begin{align*}
  \QU{E}(\QU{X})
  :=
  \bigcup_{ \rho \in \QU{X} }
  \QU{E}(\rho).
\end{align*}
In addition,
we refer to
an error where $s$-deletion errors and $t$ generalized insertion errors occur in combination
as a generalized $(s,t)$-insdel error.
For example,
a generalized $(1,2)$-insdel error includes
\begin{align*}
  &\QU{D}^1 \circ \QU{I}^{\text{G},2},
  &&\QU{I}^{\text{G},1} \circ \QU{D}^1 \circ \QU{I}^{\text{G},1} ,
  &&\QU{I}^{\text{G},2} \circ \QU{D}^1.
\end{align*}
In general,
these compositions are each distinct \cite{hagiwara2025introduction}.

\subsubsection{Inclusion Relation for Sets of States After Errors}
This section
proves that any generalized $(s,t)$-insdel error is included in $\QU{I}^{\text{G},t} \circ \QU{D}^s$.
\Corr{The following lemma is a key lemma for the proof of this property.}

\Corr{\begin{lem}
  \label{lem:1ins_1del_kai}
  Any set $\QU{X} \subset \state{l}{n}$ satisfies
  \begin{align*}
    \QU{D}^1 \circ \QU{I}^{\text{G},1}(\QU{X})
    \subset
    \QU{I}^{\text{G},1} \circ \QU{D}^1(\QU{X}).
  \end{align*}
\end{lem}}

\Corr{To prove this lemma, 
we use the following lemma.}

\begin{lem}
  \label{lem_ID=non}
  For $t \in \N$ and $P \in \subint{n}{t}$,
  $\QU{X} \subset \state{l}{n}$ satisfies
  \begin{align*}
    \QU{X}
    \subset
    \QU{I}_P^{\text{G}} \circ \QU{D}_P(\QU{X}).
  \end{align*} 
\end{lem}

\begin{IEEEproof}
  Since $\rho \in \QU{X} \subset \state{l}{n}$ satisfies
  \begin{align*}
    \QU{I}_P^{\text{G}} \circ \QU{D}_P(\{\rho\})
    &=
    \QU{I}_P^{\text{G}}( \QU{D}_P( \rho ) )\\
    &=
    \{
    \pi \in \state{l}{n}
    \mid
    \QU{D}_P(\pi) = \QU{D}_P(\rho)
    \}
    \ni
    \rho,
  \end{align*}
  we get
  $\QU{X} \subset \QU{I}_P^{\text{G}} \circ \QU{D}_P(\QU{X})$.
\end{IEEEproof}

\begin{IEEEproof}[Proof of Lemma~\ref{lem:1ins_1del_kai}]
  To simplify the notation,
  we write $\QU{D}_{\{ p \}}$ as $\QU{D}_p$
  and
  $\QU{I}_{\{ q \}}^{\text{G}}$ as $\QU{I}_q^{\text{G}}$.
  It suffices to prove that
  for $\rho \in \state{l}{n}$,
  \begin{align}
    \forall
    p,q \in \setint{n+1}
    ~~~
    \QU{D}_q \circ \QU{I}_p^{\text{G}}( \rho )
    \subset
    \QU{I}^{\text{G},1} \circ \QU{D}^1( \rho ).
    \label{eq:juubunn}    
  \end{align}
  
  If $n = 1$,
  Eq.~\eqref{eq:juubunn} holds
  since
  $\QU{I}^{\text{G},1} \circ \QU{D}^1( \rho ) = \state{l}{1}$.
  Hence,
  consider the case for $n \ge 2$.
  If $p = q$,
  from Lemma \ref{lem_ID=non},
  Eq.~\eqref{eq:juubunn} follows
  \begin{align*}
    \QU{D}_p \circ \QU{I}_p^{\text{G}}( \rho )
    &=
    \{ \rho \}
    \subset
    \QU{I}_p^{\text{G}} \circ \QU{D}_p( \rho ). 
  \end{align*}
  For $p < q$,
  by substituting $\QU{X} = \QU{D}_q \circ \QU{I}_p^{\text{G}}( \rho )$ into Lemma \ref{lem_ID=non},
  we have 
  \begin{align*}
    \QU{D}_q \circ \QU{I}_p^{\text{G}}( \rho )
    &\subset
    \QU{I}_p^{\text{G}} \circ \QU{D}_p \circ \QU{D}_q \circ \QU{I}_p^{\text{G}}( \rho )\\
    &=
    \QU{I}_p^{\text{G}} \circ \QU{D}_{q-1} \circ \QU{D}_p \circ \QU{I}_p^{\text{G}}( \rho )\\
    &=
    \QU{I}_p^{\text{G}} \circ \QU{D}_{q-1}( \rho )\\
    &\subset
    \QU{I}^{\text{G},1} \circ \QU{D}^1( \rho ),
  \end{align*}
  where
  the first equality
  follows
  from $\QU{D}_p \circ \QU{D}_q = \QU{D}_{q-1} \circ \QU{D}_p$.
  Similarly,
  we get Eq.~\eqref{eq:juubunn}
  for $p > q$.
  Thus,
  we get Eq.~\eqref{eq:juubunn}.
\end{IEEEproof}

\begin{exam}
  \label{exam:swap_TinsTdel}
  Define $\sigma_1$, $\sigma_2$, $\sigma_3$ and $\rho$ as in
  Example \ref{exam:ins_set}.
  Then,
  we get
  \begin{align*}
    \QU{D}_{\{1\}}(\sigma_1)
    =
    \QU{D}_{\{2\}}(\sigma_2)
    &=
    \QU{D}_{\{3\}}(\sigma_3)
    =
    \rho,\\
    \QU{D}_{\{2\}}(\sigma_1)
    =
    \QU{D}_{\{3\}}(\sigma_1)
    &=
    \QU{D}_{\{1\}}(\sigma_2)\\
    &=
    p_0 ~ \pi_{00} \otimes \ket{0}\bra{0}
    +
    p_1 ~ \pi_{11} \otimes \ket{1}\bra{1},\\
    \QU{D}_{\{3\}}(\sigma_2)
    =
    \QU{D}_{\{1\}}(\sigma_3)
    &=
    \QU{D}_{\{2\}}(\sigma_3)\\
    &=
    p_0 \ket{0}\bra{0} \otimes \pi_{00}
    +
    p_1 \ket{1}\bra{1} \otimes \pi_{11}.
  \end{align*}
  Hence,
  \begin{align*}
    &\QU{D}^1 \circ \QU{I}^{\text{G},1}(\rho)\\
    &=
    \{ \rho \}\\
    &~~~\cup
    \left\{
    p_0 \pi_{00} \otimes \ket{0}\bra{0}
    +
    p_1 \pi_{11} \otimes \ket{1}\bra{1}
    \mid
    \pi_{00}, \pi_{11} \in S(\C^2)
    \right\}\\
    &~~~\cup
    \left\{
    p_0 \ket{0}\bra{0} \otimes \pi_{00}
    +
    p_1 \ket{1}\bra{1} \otimes \pi_{11}
    \mid
    \pi_{00}, \pi_{11} \in S(\C^2)
    \right\}
  \end{align*}
  Since
  $\QU{D}^1(\rho) = \{ p_0 \ket{0}\bra{0} + p_1 \ket{1}\bra{1} \}$,
  from Corollary \ref{coro:insset},
  we have
  \begin{align*}
    \QU{I}_{\{1\}}^{\text{G}} \circ &\QU{D}^1(\rho)\\
    =
    \Big\{
    &p_0 ~ \pi_{00} \otimes \ket{0}\bra{0}
    +p_1 ~ \pi_{11} \otimes \ket{1}\bra{1}\\
    &+ \sqrt{p_1p_0} \mat{A} \otimes \ket{1}\bra{0}
    + \sqrt{p_0p_1} \mat{A}^\dagger \otimes \ket{0}\bra{1} \ge 0\\
    &\mid
    \pi_{00}, \pi_{11} \in \state{2}{1},
    \mat{A} \in M(2),
    \text{tr}(\mat{A}) = 0
    \Big\},\\
    \QU{I}_{\{2\}}^{\text{G}} \circ &\QU{D}^1(\rho)\\
    =
    \Big\{
    &p_0 \ket{0}\bra{0} \otimes \pi_{00}
    +p_1 \ket{1}\bra{1} \otimes \pi_{11}\\
    &+ \sqrt{p_1p_0} \ket{1}\bra{0} \otimes \mat{A}
    + \sqrt{p_0p_1} \ket{0}\bra{1} \otimes \mat{A}^\dagger \ge 0\\
    &\mid
    \pi_{00}, \pi_{11} \in \state{2}{1},
    \mat{A} \in M(2),
    \text{tr}(\mat{A}) = 0
    \Big\}.
  \end{align*}
  From the above,
  $\QU{D}^1 \circ \QU{I}^{\text{G},1}(\rho) \subset \QU{I}^{\text{G},1} \circ \QU{D}^1(\rho)$ holds.
  When $p_0, p_1 \neq 0$,
  $\QU{D}^1 \circ \QU{I}^{\text{G},1}(\rho) \neq \QU{I}^{\text{G},1} \circ \QU{D}^1(\rho)$ holds
  since
  \begin{align*}
    \ket{\psi}
    :=
    \sqrt{p_0} 
    \ket{01}
    +
    \sqrt{p_1} 
    \ket{10}
    \in
    \QU{I}^{\text{G},1} \circ \QU{D}^1(\rho) \setminus \QU{D}^1 \circ \QU{I}^{\text{G},1}(\rho).
  \end{align*} 
\end{exam}

\begin{rem}
  For some $\rho$,
  $\QU{D}^1 \circ \QU{I}^{\text{G},1}(\rho) = \QU{I}^{\text{G},1} \circ \QU{D}^1(\rho)$ holds.
  For example,
  $\rho = \ket{00}\bra{00}$ satisfies
  \begin{align*}
    &\QU{D}^1 \circ \QU{I}^{\text{G},1}(\rho)
    =
    \QU{I}^{\text{G},1} \circ \QU{D}^1(\rho)\\
    &=
    \bigcup_{\pi \in S(\C^2)}
    \{
    \ket{0}\bra{0} \otimes \pi,
    \pi \otimes \ket{0}\bra{0} \}.
  \end{align*}
\end{rem}

\Corr{Lemma \ref{lem:1ins_1del_kai} means that
shifting a generalized insertion error to the left
enlarges the set of states after the errors.
By shifting the errors as in \cite{shibayama2021equivalence},
Lemma~\ref{lem:1ins_1del_kai} is extended
to multiple deletion and insertion errors as follows:}

\begin{lem}
  \label{lem:swap_TinsTdel_kai}
  For $s, t \in \N$,
  $\QU{X} \subset \state{l}{n}$ satisfies
  \begin{align*}
    \QU{D}^s \circ \QU{I}^{\text{G},t} (\QU{X})
    \subset
    \QU{I}^{\text{G},t} \circ \QU{D}^s(\QU{X}).
  \end{align*}
\end{lem}

\begin{IEEEproof}
  Since
  the statement for $s = 0$ or $t = 0$ 
  is trivial,
  assume $s,t \ge 1$.
  From Lemma \ref{lem:1ins_1del_kai},
  we get
  \begin{align*}
    \QU{D}^{s} \circ \QU{I}^{\text{G},t} (\QU{X})
    &=
    \underbrace{
      \QU{D}^{1} \circ \cdots \circ \QU{D}^{1} \circ \B{\QU{D}^{1}}
    }_{s~\text{times}}
    \circ 
    \R{\QU{I}^{\text{G},1}} \circ \QU{I}^{\text{G},t-1}  
    (\QU{X})\\
    &\subset
    \QU{D}^{1} \circ \cdots \circ \QU{D}^{1} \circ \R{\QU{I}^{\text{G},1}}
    \circ 
    \B{\QU{D}^{1}} \circ \QU{I}^{\text{G},t-1} 
    (\QU{X})\\
    &\subset
    \QU{D}^{1} \circ \cdots \circ \R{\QU{I}^{\text{G},1}} \circ \QU{D}^{1}
    \circ 
    \QU{D}^{1} \circ \QU{I}^{\text{G},t-1} 
    (\QU{X})\\
    &~~~~~\vdots\\
    &\subset
    \R{\QU{I}^{\text{G},1}} \circ \cdots \circ \QU{D}^{1} \circ \QU{D}^{1} \circ \QU{D}^{1}
    \circ 
    \QU{I}^{\text{G},t-1} 
    (\QU{X})\\
    &=
    \R{\QU{I}^{\text{G},1}} \circ \QU{D}^{s}
    \circ 
    \R{\QU{I}^{\text{G},t-1}}
    (\QU{X}).
  \end{align*}
  Similarly,
  \begin{align*}
    \QU{D}^{s} \circ \QU{I}^{\text{G},t}(\QU{X})
    &\subset
    \QU{I}^{\text{G},1} \circ \QU{D}^{s} \circ \QU{I}^{\text{G},t-1}
    (\QU{X})\\
    &\subset
    \QU{I}^{\text{G},2} \circ \QU{D}^{s} \circ \QU{I}^{\text{G},t-2}
    (\QU{X})\\
    &~~~~~\vdots\\
    &\subset
    \QU{I}^{\text{G},t} \circ \QU{D}^{s}(\QU{X})
  \end{align*}
  holds.
\end{IEEEproof}

The following theorem shows that 
any generalized $(s,t)$-insdel error is included in $\QU{I}^{\text{G},t} \circ \QU{D}^s$.
In other words,
if a quantum code corrects $\QU{I}^{\text{G},t} \circ \QU{D}^s$,
then it also corrects any generalized $(s,t)$-insdel error.

\begin{theo}
  \label{theo:insdel_sphere}
  For $\QU{X} \subset \state{l}{n}$,
  let $\QU{A}$ be
  the set of quantum states after generalized $(s,t)$-insdel errors.
  Then,
  $\QU{A} \subset \QU{I}^{\text{G},t} \circ \QU{D}^s(\QU{X})$ holds.
\end{theo}

\begin{IEEEproof}
  Any generalized $(s,t)$-insdel error is represented by
  \begin{align*}
    \QU{I}^{\text{G},t_1} \circ \QU{D}^{s_{1}} \circ
    \QU{I}^{\text{G},t_2} \circ \QU{D}^{s_{2}} \circ
    \cdots \circ
    \QU{I}^{\text{G},t_k} \circ \QU{D}^{s_{k}}
  \end{align*}
  for some $k \in \Z^+$,
  $(s_i) \in \QU{S}_k$,
  $(t_i) \in \QU{T}_k$,
  where
  \begin{align*}
    \QU{S}_k
    &:=
    \left\{
    (s_i)
    \in
    (\Z^+)^{k-1}
    \times
    \N
    \mid
    s = \sum_{i=1}^k s_{i}
    \right\},\\
    \QU{T}_k
    &:=
    \left\{
    (t_i)
    \in
    \N
    \times
    (\Z^+)^{k-1}
    \mid
    t = \sum_{i=1}^k t_{i}
    \right\}.
  \end{align*} 
  Define $\QU{A} := \bigcup_{k \in \Z^+} \QU{A}_k$,
  where
  \begin{align*}
    \QU{A}_k
    &:=
    \hspace*{-2mm}
    \bigcup_{\bm{s} \in \QU{S}_k, \bm{t} \in \QU{T}_k}
    \hspace*{-2mm}
    \QU{I}^{\text{G},t_1} \circ \QU{D}^{s_{1}} \circ
    \QU{I}^{\text{G},t_2} \circ \QU{D}^{s_{2}} \circ
    \cdots \circ
    \QU{I}^{\text{G},t_k} \circ \QU{D}^{s_{k}}
    (\QU{X}).
  \end{align*}    
  From Lemma \ref{lem:swap_TinsTdel_kai},
  we get
  \begin{align*}
    &\QU{A}_{k+1}\\
    &=
    \hspace*{-2mm}
    \bigcup_{\bm{s} \in \QU{S}_{k+1}, \bm{t} \in \QU{T}_{k+1}}
    \hspace*{-2mm}
    \QU{I}^{\text{G},t_1} \circ \QU{D}^{s_{1}} \circ
    \cdots \circ
    \QU{I}^{\text{G},t_k} \circ \B{\QU{D}^{s_{k}}} \circ
    \R{\QU{I}^{\text{G},t_{k+1}}} \circ \QU{D}^{s_{k+1}}
    (\QU{X})\\
    &\subset
    \hspace*{-2mm}
    \bigcup_{\bm{s} \in \QU{S}_{k+1}, \bm{t} \in \QU{T}_{k+1}}
    \hspace*{-2mm}
    \QU{I}^{\text{G},t_1} \circ \QU{D}^{s_{1}} \circ
    \cdots \circ
    \QU{I}^{\text{G},t_{k}} \circ \R{\QU{I}^{\text{G},t_{k+1}}} \circ
    \B{\QU{D}^{s_{k}}} \circ \QU{D}^{s_{k+1}}
    (\QU{X})\\
    &=
    \hspace*{-2mm}
    \bigcup_{\bm{s} \in \QU{S}_{k+1}, \bm{t} \in \QU{T}_{k+1}}
    \hspace*{-2mm}
    \QU{I}^{\text{G},t_1} \circ \QU{D}^{s_{1}} \circ
    \cdots \circ
    \QU{I}^{\text{G},(t_{k}+t_{k+1})} \circ
    \QU{D}^{(s_{k}+s_{k+1})}
    (\QU{X})\\
    &\subset
    \QU{A}_{k}
  \end{align*}
  Hence,
  we have
  \begin{align*}
    \cdots
    \subset \QU{A}_{k+1}
    \subset \QU{A}_{k}
    \subset \QU{A}_{k-1}
    \subset 
    \cdots    
    \subset \QU{A}_{1}
    \subset
    \QU{I}^{\text{G},t} \circ \QU{D}^s(\QU{X}).
  \end{align*}
  Therefore,
  $\QU{A} = \bigcup_{k \in \Z^+} \QU{A}_k \subset \QU{I}^{\text{G},t} \circ \QU{D}^s(\QU{X})$
  holds.
\end{IEEEproof}

\section{Correcting Capability for Generalized Insertion and Deletion Errors}
\label{sec:gene_indel_correct}
\subsection{Definition of Error Correction}
\Corr{The set of sequences obtained from a sequence $\bm{c}$ after classical deletion/insertion errors is called 
a {\it deletion/insertion sphere} centered at $\bm{c}$.
Classical deletion/insertion-correcting codes are defined as 
the sets of sequences whose deletion/insertion spheres are pairwise disjoint \cite{sloane2002single},
e.g.,
a classical deletion/insertion-correcting code is defined by a set $C \subset \Z_l^n$ of sequences satisfying 
\begin{align*}
  \forall \bm{c}_1, \bm{c}_2 \in C,
  i \in \setint{n},
  ~~
  \bm{c}_1 \neq \bm{c}_2
  \Rightarrow
  (\bm{c}_1)_{\neg i} \neq (\bm{c}_2)_{\neg i}.
\end{align*}}

This paper employs the following definition of the error-correcting capability
similar to classical coding theory.
\Corr{Note that,
in general, 
it does not imply the existence of a recovery operator.
In the remainder of this section,
``corrects'' is used in the sense of Definition~\ref{def:error_corrct}.}

\begin{defi}
  \label{def:error_corrct}
  For $\QU{X} \subset \state{l}{n}$ and a quantum error $\QU{E}$,
  we say that $\QU{X}$ corrects $\QU{E}$
  if $\QU{X}$ satisfies
  \begin{align*}
    \forall
    \rho_1, \rho_2
    \in
    \QU{X}
    ~
    \rho_1 \neq \rho_2
    \Rightarrow
    \QU{E}(\rho_1)
    \cap
    \QU{E}(\rho_2)
    =
    \emptyset.
  \end{align*}
\end{defi}

\subsection{Properties of Deletion and Generalized Insertion Errors}
\label{subsec:prop_indel}
This section derives the properties of deletion and generalized insertion errors.
From Definition~\ref{def:gene_ins},
we get
\begin{align}
  \rho_1 \in \QU{I}_Q^{\text{G}}(\rho_2)
  &\iff
  \QU{D}_Q(\rho_1) = \rho_2.
  \label{eq:Qins_def_iff}
\end{align}

The following lemma is
an extension of Eq.~\eqref{eq:Qins_def_iff}.
\Corr{This lemma shows that
the disjointness does not change
when a generalized insertion error acting on one set
is replaced by a deletion error acting on the other set.}

\begin{lem}
  \label{lem:sins_iff_sdel_kai}
  Any $\QU{X}_1, \QU{X}_2 \subset \state{l}{\N}$ satisfy
  \begin{align*}
    \QU{X}_1 \cap \QU{I}^{\text{G},s}(\QU{X}_2)
    \neq
    \emptyset
    \iff
    \QU{D}^s(\QU{X}_1) \cap \QU{X}_2
    \neq
    \emptyset.
  \end{align*}
\end{lem}

\begin{IEEEproof}
  \Corr{From Eq.~\eqref{eq:Qins_def_iff},
  \begin{align*}
    \rho_1 \in \QU{I}^{\text{G},s}(\rho_2)
    \iff&
    \exists Q \in \subint{m+s}{s}
    ~
    \rho_1 \in \QU{I}_Q^{\text{G}}(\rho_2)\\
    \iff&
    \exists Q \in \subint{m+s}{s}
    ~
    \QU{D}_Q(\rho_1) = \rho_2\\
    \iff&
    \QU{D}^s(\rho_1) \ni \rho_2
  \end{align*}
  holds.}
  Hence,
  \begin{align*}
    \QU{X}_1 \cap \QU{I}^{\text{G},s}(\QU{X}_2)
    \neq
    \emptyset
    \iff&
    \exists
    \rho_1 \in \QU{X}_1
    ~
    \rho_1 \in \QU{I}^{\text{G},s}(\QU{X}_2)\\
    \iff&
    \exists
    \rho_1 \in \QU{X}_1
    ~
    \exists
    \rho_2 \in \QU{X}_2
    ~
    \rho_1 \in \QU{I}^{\text{G},s}(\rho_2)\\
    \iff&
    \exists
    \rho_1 \in \QU{X}_1
    ~
    \exists
    \rho_2 \in \QU{X}_2
    ~
    \QU{D}^s(\rho_1) \ni \rho_2\\
    \iff&
    \exists
    \rho_2 \in \QU{X}_2
    ~
    \QU{D}^s(\QU{X}_1) \ni \rho_2\\
    \iff&
    \QU{D}^s(\QU{X}_1) \cap \QU{X}_2
    \neq
    \emptyset
  \end{align*}
  holds.
\end{IEEEproof}

\Corr{Since
the disjointness does not change
when a generalized insertion error acting on one set
is replaced by a deletion error acting on the other set,
the following is derived
from Lemma \ref{lem:sins_iff_sdel_kai}.}

\begin{lem}
  \label{lem:sub_set}
  Any $\QU{X}_1,\QU{X}_2 \subset \state{l}{\N}$ satisfy
  \begin{align*}
    \QU{X}_1
    \cap
    \QU{I}^{\text{G},s} \circ \QU{D}^t(\QU{X}_2)
    \neq
    \emptyset
    &\iff
    \QU{D}^s(\QU{X}_1) \cap \QU{D}^t(\QU{X}_2)
    \neq
    \emptyset,\\
    \QU{X}_1
    \cap
    \QU{D}^s \circ \QU{I}^{\text{G},t}(\QU{X}_2)
    \neq
    \emptyset
    &\iff
    \QU{I}^{\text{G},s}(\QU{X}_1) \cap \QU{I}^{\text{G},t}(\QU{X}_2)
    \neq
    \emptyset.
  \end{align*}
  In particular,
  if $\QU{X}_1 = \{ \rho \}$,
  \begin{align*}
    \rho
    \in
    \QU{I}^{\text{G},s} \circ \QU{D}^t(\QU{X}_2)
    &\iff
    \QU{D}^s(\rho) \cap \QU{D}^t(\QU{X}_2)
    \neq
    \emptyset,\\
    \rho
    \in
    \QU{D}^s \circ \QU{I}^{\text{G},t}(\QU{X}_2)
    &\iff
    \QU{I}^{\text{G},s}(\rho) \cap \QU{I}^{\text{G},t}(\QU{X}_2)
    \neq
    \emptyset
  \end{align*}
  holds.
\end{lem}

\subsection{Deletion-Correcting Codes Correct Generalized Insdel Errors}
\Corr{The preceding lemmas allow us to replace generalized insertion errors
by deletion errors in the disjointness condition.
The following lemma uses this property to replace one deletion error
with one generalized insertion error while preserving disjointness;
namely,
for a pair $(\rho_1, \rho_2)$ of quantum states,
if $(\rho_1, \rho_2)$ corrects a generalized $(s,t)$-insdel error,
then
it also corrects generalized $(s-1,t+1)$-insdel error.}

\begin{lem}
  \label{lem:1del_iff_1ins_2}
  For $s \in \Z^+$, $t \in \N$, and quantum states $\rho_1, \rho_2 \in \state{l}{n}$,
  if $\QU{I}^{\text{G},t} \circ \QU{D}^{s}(\rho_1) \cap \QU{I}^{\text{G},t} \circ \QU{D}^{s}(\rho_2) =\emptyset$,
  then
  \begin{align*}
    \QU{I}^{\text{G},(t+1)} \circ \QU{D}^{(s-1)}(\rho_1)
    \cap
    \QU{I}^{\text{G},(t+1)} \circ \QU{D}^{(s-1)}(\rho_2)
    =
    \emptyset.
  \end{align*}
\end{lem}

\begin{IEEEproof}
  From Lemmas \ref{lem:sins_iff_sdel_kai} and \ref{lem:sub_set},
  \begin{align}
    &\QU{I}^{\text{G},t} \circ \QU{D}^{s}(\rho_1)  \cap \QU{I}^{\text{G},t} \circ \QU{D}^{s}(\rho_2)
    =
    \emptyset
    \nonumber\\
    \iff&
    \QU{D}^{s}(\rho_1)  \cap \QU{D}^{t} \circ \QU{I}^{\text{G},t} \circ \QU{D}^{s}(\rho_2)
    =
    \emptyset
    \nonumber\\
    \iff&
    \rho_1
    \notin
    \QU{I}^{\text{G},s} \circ \QU{D}^{t} \circ \QU{I}^{\text{G},t} \circ \QU{D}^{s}(\rho_2)
    \label{eq:assu_change1}
  \end{align}
  holds.
  From Lemma \ref{lem:swap_TinsTdel_kai},
  we get
  \begin{align*}
    &\QU{I}^{\text{G},(s-1)}\circ \QU{D}^{(t+1)} \circ \QU{I}^{\text{G},(t+1)} \circ \QU{D}^{(s-1)}(\rho_2)\\
    &\subset
    \QU{I}^{\text{G},s} \circ \QU{D}^{(t+1)} \circ \QU{I}^{\text{G},t} \circ \QU{D}^{(s-1)}(\rho_2)\\
    &\subset 
    \QU{I}^{\text{G},s} \circ \QU{D}^{t} \circ \QU{I}^{\text{G},t} \circ \QU{D}^{s}(\rho_2).
  \end{align*}
  Hence,
  we have
  \begin{align}
    &\rho_1 \notin \QU{I}^{\text{G},s} \circ \QU{D}^{t} \circ \QU{I}^{\text{G},t} \circ \QU{D}^{s}(\rho_2)
    \nonumber\\
    &\Rightarrow
    \rho_1 \notin \QU{I}^{\text{G},(s-1)}\circ \QU{D}^{(t+1)} \circ \QU{I}^{\text{G},(t+1)} \circ \QU{D}^{(s-1)}(\rho_2).
    \label{eq:st->s-1t+1}
  \end{align}
  By a similar way to Eq.~\eqref{eq:assu_change1},
  \begin{align}
    &\rho_1
    \notin
    \QU{I}^{\text{G},(s-1)} \circ \QU{D}^{(t+1)} \circ \QU{I}^{\text{G},(t+1)} \circ \QU{D}^{(s-1)}(\rho_2)
    \nonumber\\
    \iff
    &\QU{I}^{\text{G},(t+1)} \circ \QU{D}^{(s-1)}(\rho_1)
    \cap
    \QU{I}^{\text{G},(t+1)} \circ \QU{D}^{(s-1)}(\rho_2)
    =
    \emptyset
    \label{eq:assu_change2}
  \end{align}
  holds.
  Combining
  Eqs.~\eqref{eq:assu_change1}, \eqref{eq:st->s-1t+1}, and \eqref{eq:assu_change2},
  we get the lemma.
\end{IEEEproof}

\begin{theo}
  \label{theo:1del_iff_1ins_2}
  If $\QU{X} \subset \state{l}{n}$ corrects $t$-deletion errors,
  then $\QU{X}$ also corrects a total of at most $t$ deletion and generalized insertion errors.
\end{theo}

\Corr{\begin{IEEEproof}
  Let $\rho_1, \rho_2 \in \QU{X}$ be distinct states.
  From Lemma~\ref{lem:1del_iff_1ins_2},
  \begin{align*}
    &\QU{D}^{(t)}(\rho_1)
    \cap
    \QU{D}^{(t)}(\rho_2)
    =
    \emptyset\\
    &\iff
    \QU{I}^{\text{G},(0)} \circ \QU{D}^{(t)}(\rho_1)
    \cap
    \QU{I}^{\text{G},(0)} \circ \QU{D}^{(t)}(\rho_2)
    =
    \emptyset\\
    &\Rightarrow
    \QU{I}^{\text{G},(1)} \circ \QU{D}^{(t-1)}(\rho_1)
    \cap
    \QU{I}^{\text{G},(1)} \circ \QU{D}^{(t-1)}(\rho_2)
    =
    \emptyset\\
    &\Rightarrow
    \QU{I}^{\text{G},(2)} \circ \QU{D}^{(t-2)}(\rho_1)
    \cap
    \QU{I}^{\text{G},(2)} \circ \QU{D}^{(t-2)}(\rho_2)
    =
    \emptyset\\
    &~~~\vdots\\
    &\Rightarrow
    \QU{I}^{\text{G},(t)} \circ \QU{D}^{(0)}(\rho_1)
    \cap
    \QU{I}^{\text{G},(t)} \circ \QU{D}^{(0)}(\rho_2)
    =
    \emptyset
  \end{align*}  
  holds.
  By Theorem~\ref{theo:insdel_sphere},
  the set of states after any generalized $(t-i,i)$-insdel error
  is contained in
  $\QU{I}^{\mathrm G,i} \circ \QU{D}^{t-i}$.
  Therefore, 
  $\QU{X}$ corrects every generalized $(t-i,i)$-insdel error.
  Since 
  $\QU{X}$ corrects $t'$ deletion errors
  if $\QU{X}$ corrects $t$ deletion errors ($t' \leq t$),
  $\QU{X}$ corrects a total of at most $t$
  deletion and generalized insertion errors.  
\end{IEEEproof}}

The converse of Theorem \ref{theo:1del_iff_1ins_2} does not necessarily hold,
as shown by the following example.

\begin{exam}
  Define
  $\rho$ and $\ket{\psi}$ as in Examples \ref{exam:ins_set} and \ref{exam:swap_TinsTdel},
  respectively.
  Then,
  we get
  \begin{align*}
    &\ket{\psi}\bra{\psi}
    \in
    \QU{I}^{\text{G},1} \circ \QU{D}^1(\rho),
    &&\ket{\psi}\bra{\psi}
    \notin
    \QU{D}^1 \circ \QU{I}^{\text{G},1}(\rho).
  \end{align*}
  \Corr{Note that
  the above result is given in \cite{hagiwara2025introduction}.}
  From Lemma \ref{lem:sub_set},
  we get
  \begin{align*}
    &\QU{D}^1(\ket{\psi}\bra{\psi}) \cap \QU{D}^1(\rho)
    \neq
    \emptyset,
    &&\QU{I}^{\text{G},1}(\ket{\psi}\bra{\psi}) \cap \QU{I}^{\text{G},1}(\rho)
    =
    \emptyset.
  \end{align*}
  In other words,
  $\{ \rho,\ket{\psi}\bra{\psi} \}$ corrects a single generalized insertion error,
  but does not correct a single deletion error.
\end{exam}

\section{Quantum Insdel Distance}
\label{sec:Qindel_dis}
This section defines the quantum insdel distance.
\Corr{Then,
this section characterizes} the error-correcting capability of quantum codes 
for deletion and \Corr{generalized insertion} errors using this distance.

\subsection{Definition and Examples}
The insdel distance for $\bm{x} \in \Z_l^n$ and $\bm{y} \in \Z_l^m$ is
defined as the total number of classical deletion and insertion errors
required to transform $\bm{x}$ into $\bm{y}$ \cite{levenshtein1966binary}.
Similar to the insdel distance,
we define
the quantum insdel distance for $\rho_1, \rho_2 \in \state{l}{\N}$ as
the total number of deletion and generalized insertion errors
required to transform
$\rho_1$ into $\rho_2$.
The quantum insdel distance is mathematically defined as follows.

\begin{defi}
  Let $\SET{E}$ be the set of generalized $(s,t)$-insdel errors. 
  The quantum insdel distance $d( \rho_1, \rho_2 )$ of $\rho_1, \rho_2 \in \state{l}{\N}$
  is defined as follows.
  \begin{align*}
    d( \rho_1, \rho_2 )
    &:=
    \min_{(s,t) \in \N^2, \QU{E} \in \SET{E}}
    \{
    s+t
    \mid
    \rho_2
    \in
    \QU{E}(\rho_1)
    \}\\
    &=
    \min_{(s,t) \in \N^2}
    \{
    s+t
    \mid
    \rho_2
    \in
    \QU{I}^{\text{G},t} \circ \QU{D}^s(\rho_1)
    \}\\
    &=
    \min_{(s,t) \in \N^2}
    \{
    s+t
    \mid
    \QU{D}^s(\rho_1) \cap \QU{D}^t(\rho_2)
    \neq
    \emptyset
    \}
  \end{align*}
\end{defi}

The first equality
follows
from Theorem \ref{theo:insdel_sphere}
and
the second equality
follows
from Lemma \ref{lem:sub_set}.

\begin{exam}
  Suppose
  \begin{align*}
    \rho_1
    &=
    \tfrac{1}{2}
    \ket{00}\bra{00}
    +
    \tfrac{1}{2}
    \ket{11}\bra{11},\\
    \rho_2
    &=
    \tfrac{1}{2}
    \ket{10}\bra{10}
    +
    \tfrac{1}{2}
    \ket{01}\bra{01}.
  \end{align*}
  Then,
  $d( \rho_1, \rho_2 ) = 2$ holds
  since
  \begin{align*}
    \QU{D}_{\{ 1 \}}(\rho_1)
    =
    \QU{D}_{\{ 2 \}}(\rho_2)
    =
    \tfrac{1}{2}\ket{0}\bra{0}
    +
    \tfrac{1}{2}\ket{1}\bra{1}.
  \end{align*} 
\end{exam}

\begin{defi}
  The minimum distance $d_{\text{min}}(\QU{X})$ of $\QU{X} \subset \state{l}{\N}$ with $|\QU{X}| \ge 2$
  is defined as follows:
  \begin{align*}
    d_{\text{min}}(\QU{X})
    :=
    \min_{ \rho_1, \rho_2 \in \QU{X} : \rho_1 \neq \rho_2 }
    d( \rho_1, \rho_2 )
  \end{align*}
\end{defi}

Now,
we will calculate the minimum distances of several quantum codes.
Let $j$ be the imaginary unit.

\begin{exam}
  Define
  $\QU{X}_1 := \{ \alpha \ket{00} + \beta \ket{11} \mid \alpha,\beta \in \C, |\alpha|^2+|\beta|^2 = 1 \}$.
  Note that
  $\ket{\psi} =  \alpha \ket{00} + \beta \ket{11} \in \QU{X}_1$ satisfies
  \begin{align*}
    \ket{\psi}\bra{\psi} 
    &=
    |\alpha|^2 \ket{00}\bra{00} + \alpha\beta^* \ket{00}\bra{11}\\
    &~~~+
    \beta\alpha^* \ket{11}\bra{00} + |\beta|^2 \ket{11}\bra{11}
  \end{align*}
  for any $\alpha, \beta$.
  We have
  \begin{align*}
    \QU{D}_{\{1\}}(\ket{\psi}\bra{\psi} )
    =
    \QU{D}_{\{2\}}(\ket{\psi}\bra{\psi} )
    =
    |\alpha|^2 \ket{0}\bra{0} + |\beta|^2 \ket{1}\bra{1}.
  \end{align*}
  Hence,
  \Corr{\begin{align*}
    \ket{\psi_1}
    &:=
    \tfrac{1}{\sqrt{2}}\ket{00} + \tfrac{1}{\sqrt{2}}\ket{11} \in \QU{X}_1,\\
    \ket{\psi_2}
    &:=
    \tfrac{1}{\sqrt{2}}\ket{00} - \tfrac{1}{\sqrt{2}}\ket{11} \in \QU{X}_1
  \end{align*}}
  satisfy
  \begin{align*}
    \tfrac{1}{2} \ket{0}\bra{0} + \tfrac{1}{2} \ket{1}\bra{1}
    \in
    \QU{D}^1(\ket{\psi_1}\bra{\psi_1})
    \cap
    \QU{D}^1(\ket{\psi_2}\bra{\psi_2}).
  \end{align*}
  Since the quantum insdel distance between states of the same length is even,
  $d_{\text{min}}(\QU{X}_1) = 2$ holds.
\end{exam}
  
\begin{exam}
  \label{exam:dis=4}
  \Corr{For $\bm{x}\in\Z_2^n$, let $\mathrm{wt}(\bm{x})$ denote
  the Hamming weight of $\bm{x}$.
  The Dicke state of weight $i$ is defined as
  $\ket{D_i^n} := \sum_{\bm{x} \in \Z_2^n : \text{wt}(\bm{x})=i} \ket{\bm{x}} / \sqrt{\binom{n}{i}}$.}
  Consider
  the minimum distance of the following quantum single-deletion-correcting code
  $\QU{X}_2$ given in \cite{hagiwara2020four}:
  \begin{align*}
    &\QU{X}_2
    :=
    \left\{
    \alpha
    \ket{0_L}
    +
    \beta
    \ket{1_L}
    \mid
    \alpha,\beta \in \C, |\alpha|^2+|\beta|^2 = 1
    \right\},
  \end{align*}
  where
  \Corr{\begin{align*} 
    &\ket{0_L}
    :=
    \tfrac{ 1 }{ \sqrt{2} }\big( \ket{D_0^4} + \ket{D_4^4}\big),
    ~~~~
    \ket{1_L}
    :=
    \ket{D_2^4}.
  \end{align*}}
  For $p \in \setint{4}$ and
  $\ket{\psi} = \alpha \ket{0_L} + \beta \ket{1_L} \in \QU{X}_2$,
  we have
  \begin{align}
    &\QU{D}_{\{p\}}(\ket{\psi}\bra{\psi})
    \nonumber\\
    &=
    \tfrac{1}{2}
    \Corr{\big( \alpha\ket{D_0^3} + \beta\ket{D_2^3} \big)
    \big( \alpha^*\bra{D_0^3} + \beta^*\bra{D_2^3} \big)}
    \nonumber\\
    &~~~+
    \tfrac{1}{2}
    \Corr{\big( \alpha\ket{D_3^3}+\beta\ket{D_1^3} \big)
    \big( \alpha^*\bra{D_3^3}+\beta^*\bra{D_1^3} \big).}
    \label{eq:expm_del_code_dis}
  \end{align}
  For distinct $\ket{\psi_1}, \ket{\psi_2} \in \QU{X}_2$,
  $\QU{D}^1(\ket{\psi_1}\bra{\psi_1}) \cap \QU{D}^1(\ket{\psi_2}\bra{\psi_2}) = \emptyset$
  since
  \Corr{$\ket{D_0^3}, \ket{D_1^3}, \ket{D_2^3}, \ket{D_3^3}$}
  are
  mutually orthogonal.
  That is,
  $d_{\text{min}}(\QU{X}_2) \ge 3$.
  From Eq.~\eqref{eq:expm_del_code_dis},
  for $P \in \subint{4}{2}$,
  we have
  \begin{align*}
    \QU{D}_P(\ket{\psi}\bra{\psi})
    =&
    \tfrac{1}{2}
    \big(
    |\alpha|^2 + \tfrac{|\beta|^2}{3}
    \big)
    \Corr{\big(
    \ket{D_0^2}\bra{D_0^2} + \ket{D_2^2}\bra{D_2^2}
    \big)}\\
    &+
    \tfrac{ \alpha\beta^* + \alpha^*\beta }{2\sqrt{3}}
    \Corr{\big(
    \ket{D_0^2}\bra{D_2^2} + \ket{D_2^2}\bra{D_0^2}
    \big)}\\
    &+
    \tfrac{2|\beta|^2}{3}\Corr{\ket{D_1^2}\bra{D_1^2},}
  \end{align*}
  for any $\alpha, \beta$.
  Then,
  \Corr{\begin{align*}
    \ket{\psi_1}
    &:=
    \tfrac{1}{\sqrt{2}} \ket{0_L} + \tfrac{j}{\sqrt{2}} \ket{1_L} \in \QU{X}_2,\\
    \ket{\psi_2}
    &:=
    \tfrac{1}{\sqrt{2}} \ket{0_L} - \tfrac{j}{\sqrt{2}} \ket{1_L} \in \QU{X}_2
  \end{align*}}
  satisfy
  \begin{align*}
    &\frac{1}{3}
    \Corr{( \ket{D_0^2}\bra{D_0^2} + \ket{D_1^2}\bra{D_1^2} + \ket{D_2^2}\bra{D_2^2} )}\\
    &\in
    \QU{D}^2(\ket{\psi_1}\bra{\psi_1})
    \cap
    \QU{D}^2(\ket{\psi_2}\bra{\psi_2}).
  \end{align*}
  Hence,
  $d_{\text{min}}(\QU{X}_2) = 4$.
\end{exam}

\Corr{\begin{exam}
  \label{exam:dis=4_non_PI}
  Consider the minimum distance of the following quantum single-deletion-correcting code $\QU{X}_3$ \cite{shibayama2022construction}:
  \begin{align*}
    &\QU{X}_3
    :=
    \left\{
    \alpha
    \ket{0_L}
    +
    \beta
    \ket{1_L}
    \mid
    \alpha,\beta \in \C, |\alpha|^2+|\beta|^2 = 1
    \right\},
  \end{align*}
  where
  \begin{align*} 
    \ket{0_L}
    &:=
    \tfrac{ 1 }{ \sqrt{2} }\big( \ket{000100} + \ket{110111}\big),\\
    \ket{1_L}
    &:=
    \tfrac{ 1 }{ \sqrt{2} }\big( \ket{000111} + \ket{110100}\big).
  \end{align*}
  Note that this code is not a PI code.
  Since
  \begin{align*}
    &\tfrac{1}{2}
    \big(
      \ket{0100}\bra{0100}
      +
      \ket{0111}\bra{0111}
    \big)
    \\
    &\in
    \QU{D}^2(\ket{0_L}\bra{0_L})
    \cap
    \QU{D}^2(\ket{1_L}\bra{1_L}),
  \end{align*}
  $d_{\mathrm{min}}(\QU{X}_3)\le 4$ holds.
  For
  $\ket{\psi} = \alpha \ket{0_L} + \beta \ket{1_L} \in \QU{X}_3$,
  we get 
  \begin{align*}
    &\QU{D}_{\{1\}}(\ket{\psi}\bra{\psi}) = \QU{D}_{\{2\}}(\ket{\psi}\bra{\psi})\\
    &=
    \tfrac{1}{2}
    \big(
      \alpha\ket{00100}+\beta\ket{00111}
    \big)
    \big(
      \alpha^*\bra{00100}+\beta^*\bra{00111}
    \big)\\
    &~~~+
    \tfrac{1}{2}
    \big(
      \alpha\ket{10111}+\beta\ket{10100}
    \big)
    \big(
      \alpha^*\bra{10111}+\beta^*\bra{10100}
    \big),\\
    &\QU{D}_{\{3\}}(\ket{\psi}\bra{\psi})\\
    &=
    \left(
    \tfrac{\alpha}{\sqrt{2}} (\ket{00100}+\ket{11111}) + \tfrac{\beta}{\sqrt{2}} (\ket{00111}+\ket{11100})
    \right)\\
    &~~~
    \left(
    \tfrac{\alpha^*}{\sqrt{2}} (\bra{00100}+\bra{11111}) + \tfrac{\beta^*}{\sqrt{2}} (\bra{00111}+\bra{11100})
    \right),\\
    &\QU{D}_{\{4\}}(\ket{\psi}\bra{\psi})\\
    &=
    \left(
    \tfrac{\alpha}{\sqrt{2}} (\ket{00000}+\ket{11011}) + \tfrac{\beta}{\sqrt{2}} (\ket{00011}+\ket{11000})
    \right)\\
    &~~~
    \left(
    \tfrac{\alpha^*}{\sqrt{2}} (\bra{00000}+\bra{11011}) + \tfrac{\beta^*}{\sqrt{2}} (\bra{00011}+\bra{11000})
    \right),\\
    &\QU{D}_{\{5\}}(\ket{\psi}\bra{\psi})
    =
    \QU{D}_{\{6\}}(\ket{\psi}\bra{\psi})\\
    &=
    \tfrac{1}{2}
    \big(
      \alpha\ket{00010}+\beta\ket{11010}
    \big)
    \big(
      \alpha^*\bra{00010}+\beta^*\bra{11010}
    \big)\\
    &~~~+
    \tfrac{1}{2}
    \big(
      \alpha\ket{11011}+\beta\ket{00011}
    \big)
    \big(
      \alpha^*\bra{11011}+\beta^*\bra{00011}
    \big).
  \end{align*}
  Hence, 
  for distinct $\ket{\psi_1}, \ket{\psi_2} \in \QU{X}_3$,
  $\QU{D}^1(\ket{\psi_1}\bra{\psi_1}) \cap \QU{D}^1(\ket{\psi_2}\bra{\psi_2}) = \emptyset$ holds.
  Thus,
  we get
  $d_{\mathrm{min}}(\QU{X}_3)\ge 3$
  and
  $d_{\mathrm{min}}(\QU{X}_3) = 4$.
\end{exam}}

\subsection{Properties}
\begin{theo}
  The quantum insdel distance is a metric on $\state{l}{\N}$.
\end{theo}

\begin{IEEEproof}
  Any $\rho_1, \rho_2 \in \state{l}{\N}$ satisfy
  \begin{align*}
    d(\rho_1, \rho_2) = 0
    &\iff
    \emptyset
    \neq
    \QU{D}^0(\rho_1) \cap \QU{D}^0(\rho_2)
    =
    \{ \rho_1 \} \cap \{ \rho_2 \}\\
    &\iff
    \rho_1 = \rho_2.
  \end{align*}
  The non-negativity and the symmetry are trivial.
  Now,
  we will prove the triangle inequality.
  For $\rho_1, \rho_2, \rho_3 \in \state{l}{\N}$,
  define
  \begin{align*}
    (s_1,t_1)
    &:=
    \underset{(s,t) \in \N^2}{\operatorname{argmin}}~
    \{
    s+t
    \mid
    \rho_2 \in \QU{I}^{\text{G},t} \circ \QU{D}^{s}(\rho_1)
    \},\\
    (s_2,t_2)
    &:=
    \underset{(s,t) \in \N^2}{\operatorname{argmin}}~
    \{
    s+t
    \mid
    \rho_2 \in \QU{I}^{\text{G},t} \circ \QU{D}^{s}(\rho_3)
    \}.
  \end{align*}
  From $\rho_2 \in \QU{I}^{\text{G},t_1} \circ \QU{D}^{s_1}(\rho_1)$ and
  $\rho_2 \in \QU{I}^{\text{G},t_2} \circ \QU{D}^{s_2}(\rho_3)$,
  we get
  \begin{align*}
    \QU{I}^{\text{G},t_1} \circ \QU{D}^{s_1}(\rho_1)
    \cap
    \QU{I}^{\text{G},t_2} \circ \QU{D}^{s_2}(\rho_3)
    \neq
    \emptyset.
  \end{align*}
  From Lemma \ref{lem:sub_set} and Theorem \ref{theo:insdel_sphere},
  we have
  \begin{align*}
    \rho_1 
    &\in
    \QU{I}^{\text{G},s_1} \circ \QU{D}^{t_1} \circ \QU{I}^{\text{G},(t_2)} \circ \QU{D}^{s_2}(\rho_3)
    \subset
    \QU{I}^{\text{G},(s_1+t_2)} \circ \QU{D}^{t_1+s_2}(\rho_3).    
  \end{align*}
  Therefore,
  \begin{align*}
    d(\rho_1 ,\rho_3)
    &\leq
    s_1+t_2 + t_1+s_2\\
    &=
    d(\rho_1, \rho_2) + d(\rho_2, \rho_3).
  \end{align*}
  holds.
\end{IEEEproof}

\begin{lem}
  \label{lem:1del_iff_1ins_1}
  For $0 \leq t \leq n$,
  any $\rho_1, \rho_2 \in \state{l}{n}$ satisfy
  \begin{align*}
    d(\rho_1,\rho_2)
    \ge
    2t+1
    \iff
    \QU{D}^t(\rho_1) \cap \QU{D}^t(\rho_2) = \emptyset
  \end{align*}
\end{lem}

\Corr{\begin{IEEEproof}
  Clearly,
  if $\QU{D}^t(\rho_1) \cap \QU{D}^t(\rho_2) \neq \emptyset$,
  $d(\rho_1,\rho_2) \leq 2t$ holds.
  Conversely,
  if $d(\rho_1,\rho_2)\le 2t$ holds,
  there exist $s_1,s_2 \in \N$ 
  such that $s_1 + s_2 \leq 2t$ and $\QU{D}^{s_1}(\rho_1) \cap \QU{D}^{s_2}(\rho_2) \neq \emptyset$.
  Since $\rho_1$ and $\rho_2$ have the same length,
  we get $n-s_1 = n-s_2$
  and
  $s_1 = s_2 \leq t$.
  Hence,
  $\QU{D}^t(\rho_1) \cap \QU{D}^t(\rho_2) \neq \emptyset$ holds.
\end{IEEEproof}}

Lemma \ref{lem:1del_iff_1ins_1} yields the following theorem.

\begin{theo}
  \label{theo:1del_iff_1ins_1}
  For $0 \leq t \leq n$,
  the set $\QU{X} \subset \state{l}{n}$
  corrects
  $t$-deletion errors
  iff $d_{\text{min}}(\QU{X}) \ge 2t+1$.
\end{theo}

\Corr{From Theorems \ref{theo:1del_iff_1ins_2} and \ref{theo:1del_iff_1ins_1},
we obtain the following.
In fact,
the quantum single-deletion-correcting codes $\QU{X}_2$ and $\QU{X}_3$ \cite{hagiwara2020four,shibayama2022construction} 
in Examples~\ref{exam:dis=4} and \ref{exam:dis=4_non_PI} satisfy
$d_{\text{min}}(\QU{X}_2) = d_{\text{min}}(\QU{X}_3) = 4$.
These codes
correct
a single insertion error,
as shown in \cite{hagiwara2021four,shibayama2022construction}.}

\begin{coro}
  For $0 \leq t \leq n$,
  the set $\QU{X} \subset \state{l}{n}$
  corrects
  a total of $t$ deletion and generalized insertion errors
  if $d_{\text{min}}(\QU{X}) \ge 2t+1$.
\end{coro}

\Corr{\section{Conclusion}
\label{sec:concl}
We studied
the insertion-correcting capability
for two type of insertion errors: 
insertion error with/without quantum channel requirement,
called restricted/generalized insertion errors.
For restricted insertion errors,
by showing that 
insertion errors are limited to BS insertion errors,
we proved the conjecture;
namely,
a quantum code $\QU{C}$ corrects $t$ deletion errors 
iff
$\QU{C}$ corrects a total of at insertion and deletion errors.
On the other hand,
for generalized insertion errors,
we showed that
the analogous equivalence does not hold in general.
This is due to 
the commutativity for deletion and insertion errors
and
the non-commutativity for deletion and generalized insertion errors.
Moreover,
we introduced the quantum insdel distance
and characterized
the error-correcting capability for deletion and generalized insertion errors.}

\Corr{Future work on this study is to clarify the conditions under which
generalized insertion-error correctability
implies deletion-error correctability.
Specifically,
since
the equivalence holds
if generalized insertion errors limited to BS insertion errors,
it is searching for class of insertions
that is broader than BS insertion errors but narrower than generalized insertion errors.}

\bibliographystyle{IEEEtran}
\bibliography{myref}

\begin{IEEEbiographynophoto}{Ken Nakamura}
  (Graduate Student Member)
  received the B.S.\ and M.S.\ degrees in Informatics
  from Yamaguchi University, Yamaguchi, Japan
  in 2023 and 2025, respectively.
  Since April 2025,
  he has been pursuing the Ph.D.\ degree in Informatics
  at Yamaguchi University.
  His research interests include error-correcting codes and quantum coding theory.  
  He is a graduate student member of IEEE and a student member of IEICE.
\end{IEEEbiographynophoto}

\Corr{\begin{IEEEbiographynophoto}{Kaito Sato}
  received the B.S.\ degree in Informatics 
  from Yamaguchi University, Yamaguchi, Japan in 2026. 
  Since April 2026, 
  he has been pursuing the M.S. degree in Physics
  at Yamaguchi University. 
  His research interests include error-correcting codes and quantum coding theory.
\end{IEEEbiographynophoto}}

\begin{IEEEbiographynophoto}{Takayuki Nozaki}
  (Member, IEEE)
  received the B.E.\ degree in computer science
  and the M.E.\
  and Dr.Eng.\ degrees in communication engineering
  from Tokyo Institute of Technology, Tokyo, Japan, in 2008, 2010, and 2012,
  respectively.
  From 2010 to 2013,
  he was a Research Fellow of Japan Society
  for the Promotion of Science.
  From 2013 to 2015,
  he was a Research Associate at Kanagawa University.
  He joined Yamaguchi University as an Assistant Professor in 2015
  and served as Senior Lecturer from 2017 to 2020.
  Since 2020,
  he has been an Associate Professor at Department of Informatics,
  Yamaguchi University, Yamaguchi, Japan.
  His research interests include error-correcting codes and coding theory.
  Dr. Nozaki is a member of IEEE, ACM, AMS, IEICE, and IPSJ.
\end{IEEEbiographynophoto}
\end{document}